\documentclass[prl, reprint, superscriptaddress, amsmath,amssymb, aps]{revtex4-2}

\usepackage{amsthm}
\newtheorem{theorem}{Theorem}
\newtheorem{corollary}{Corollary}
\newtheorem{lemma}{Lemma}
\newtheorem{proposition}{Proposition}
\newtheorem{definition}{Definition}
\newtheorem{remark}{Remark}
\usepackage{tcolorbox}
\usepackage{esint}
\usepackage{bbold}
\usepackage{amssymb}
\usepackage{graphicx}
\usepackage{textcomp}
\usepackage{relsize}
\usepackage{url}
\usepackage[colorlinks=true, allcolors=camblue]{hyperref}
\usepackage{blkarray}
\usepackage{cancel}
\usepackage{tikz}
\usepackage{natbib}
\usepackage{stmaryrd}
\usepackage[utf8]{inputenc}
\usepackage[T1]{fontenc}
\usepackage[caption=false]{subfig}
\usetikzlibrary{quantikz2}
\usepackage{multirow}
\usepackage{mathtools}
\mathtoolsset{showonlyrefs}

\definecolor{camblue}{cmyk}{0.2443, 0.0000, 0.1250, 0.3098}

\usepackage{xcolor}
\usepackage{lipsum}

\newcommand{\Hcal}{\mathcal H}
\newcommand{\Kcal}{\mathcal K}
\newcommand{\Mcal}{\mathfrak M}
\newcommand{\Tr}{\operatorname{Tr}}

\newcommand{\supp}{\operatorname{supp}}
\newcommand{\Sym}{\operatorname{Sym}}
\newcommand{\GSym}{\operatorname{GSym}}
\newcommand{\Ran}{\operatorname{Ran}}
\newcommand{\Com}{\operatorname{Com}}

\begin{document}

\title{Exponential de Finetti Theorems for Fermionic Gaussian States}
\author{Jędrzej Burkat}
\email{jbb55@cam.ac.uk}
\affiliation{Cavendish Laboratory, Department of Physics, University of Cambridge, CB3 0HE, UK}
\author{Michał Studziński}
\affiliation{International Centre for Theory of Quantum Technologies, University of Gdańsk, 80-309, Poland}
\author{Sergii Strelchuk}
\affiliation{Department of Computer Science, University of Oxford, OX1 3QD, UK}

\begin{abstract}
 We prove an exponential variant of the Gaussian de Finetti theorem: the subsystems of permutation-invariant, free-fermionic Gaussian states are well approximated by convex combinations of almost-i.i.d. states that are Gaussian on subsets of their parts. Our result provides an error bound between the original state and its approximants that decays exponentially in the number of unconstrained parts, becoming super-exponential when the subsystem under consideration is small. The dependence of our bound on the local Hilbert space dimension is polylogarithmic, giving an exponential improvement over the standard de Finetti theorem of [Nat. Phys. 3, 645--649]. Previous works considered Gaussian-symmetric states, which are supported on the trivial irrep of the tensor matchgate representation. We extend these results to the broader class of Gaussian-invariant states, which contains, for example, i.i.d. copies of single-replica mixed Gaussian states. We show that Gaussian-invariant states are precisely the partial traces of Gaussian-symmetric states on locally enlarged replicas and always admit a purification into a larger Gaussian-symmetric state. As an application, we derive explicit convergence bounds for a hierarchy of convex relaxations of polynomial optimisation over fermionic Gaussian states. For variational energy minimisation, the hierarchy gives lower bounds converging to the global Gaussian minimum. Combined with upper bounds from Gaussian trial states, these bounds quantify how far a variational calculation can be from the best energy attainable within the Gaussian family.
\end{abstract}

\maketitle

\paragraph{Introduction.} In quantum information theory, de Finetti theorems formalise the principle that when a large composite system obeys a symmetry, then its individual subsystems are largely uncorrelated or independent. In the finite-dimensional setting, both notions carry an elegant operational meaning: the canonical symmetry for a $k$-partite state $\rho$ is invariance under permutations of its parts, whereas independence is captured by tensor products of independent and identically distributed (i.i.d.) states $\sigma^{\otimes k}$. Indeed, the quantum de Finetti theorem \cite{Konig_2005, Christandl_2007,Konig_2007} guarantees that small, $k$-sized subsystems $\Tr_m(\rho^{k+m})$ of an overall $(k+m)$-partite permutation-invariant state $\rho^{k+m}$ can be expressed as a statistical mixture of i.i.d. states $\sigma^{\otimes k}$ up to a favourable error bound. This has led to a wide range of successful applications \cite{Doherty_2002, Brandao_2011, Brandao_2013, Christandl_2007}. However, areas such as quantum cryptography \cite{Renner_2007, Leverrier_2017} impose a much more stringent requirement on the precision of the estimate. Furthermore, in many physical scenarios, one may also wish to approximate a large subsystem, so as to make inferences about its global properties \cite{Lewin_2013, Rougerie_2020, Gluza_2016, Gogolin_2016}. Both of these requirements are satisfied by approximating the subsystem with almost-i.i.d. states \cite{Renner_2007}, which are only guaranteed to be i.i.d. on subsets of their parts. 

Given the success of de Finetti theorems for permutational invariance, it is natural to ask about scenarios where additional information about the symmetry is available. For example, if the overall state is invariant under a larger group $G\supseteq S_{k+m}$, then one may expect that this additional structure may lead to better approximations. This is particularly important in situations where the local subsystems become large and standard approaches struggle to provide good convergence. This has led to successes in continuous-variable systems \cite{Leverrier_2009, Krumnow_2017, Leverrier_2018}, as well as for stabilizer states \cite{Gross_2021} (see also \cite{wang_2026}). Recent work has considered scenarios where in addition to satisfying permutation invariance, the overall state is known to be free-fermionic, or Gaussian. Much of this has been made possible thanks to the newly developed theory of the matchgate commutant \cite{Sierant_2026, braccia_2026, lastres_2026}, which completely characterises the form of unitaries that leave the state invariant. Notably, it contains the symmetric group algebra $\mathbb{C}(S_k)$ as a subalgebra, and so assuming Gaussian symmetries can be seen as a generalisation of permutation invariance for free-fermionic systems.

In this article, we prove two generalisations. First, we extend the free-fermionic de Finetti theorem to approximation by almost-i.i.d. states, which are Gaussian on subsets of their parts. This yields an exponential de Finetti theorem for Gaussian-symmetric states, exponentially improving the error bound of \cite{Sierant_2026} and avoiding issues due to large local dimensions in \cite{Renner_2007}. Second, we extend the result to a larger symmetry of Gaussian-invariance, which roughly corresponds to globally Gaussian, permutation-invariant states whose constituents are not necessarily pure. Our purification theorem shows that any such state can be purified into a larger Gaussian-symmetric state, extending the applicability of Gaussian de Finetti theorems to a wider family of states. As an application of our results, we obtain explicit error bounds for convex relaxations of bounded degree optimization over fermionic Gaussian states. The fully i.i.d. estimate on $n$ qubits gives an error of $\mathcal{O}(t n^2/\ell)$, multiplied by the spectral range of an objective on $t$ replicas, at extension level $\ell$. For Gaussian variational energies, these relaxations and Gaussian trial states provide lower and upper bounds on the optimum within the Gaussian family.

\paragraph{Definitions.} We first introduce Gaussian-symmetric and Gaussian-invariant states, which form the analogue of the symmetric and permutation-invariant states for the replicated matchgate setting:

\begin{definition} \label{def:gsym}
    Let $\Hcal$ be a $2^{n}$-dimensional Hilbert space of $n$ qubits. The Gaussian-symmetric subspace of $\Hcal^{\otimes k}$, written $\GSym^k(\Hcal)$, is the space spanned by all $k$-fold tensor products (replicas) of $n$-qubit Gaussian states. Equivalently, by \cite{Sierant_2026} it is the subspace of $\Hcal^{\otimes k}$ annihilated by all the bridge operators $\Lambda_{ab}$:
    \begin{align}
    \GSym^k(\Hcal) &= \operatorname{span} \left \{ |\psi\rangle \in \Hcal^{\otimes k} : \Lambda_{ab} |\psi\rangle = 0 \ \, \forall a<b \right \} \\
    &= \operatorname{span} \left \{ (U|\mathbf{0} \rangle)^{\otimes k} \in \Hcal^{\otimes k} :  U \in \mathfrak{M}_n \right \},
    \end{align}
    where $|\mathbf{0}\rangle = |0^n \rangle$ is an $n$-qubit vacuum state, $\mathfrak{M}_n$ is the $n$-qubit matchgate group, and the bridge operators $\Lambda_{ab}$ are defined as:
    \begin{equation}
    \Lambda_{ab} = \sum_{\mu = 1}^{2n} \gamma^{(a)}_\mu \gamma^{(b)}_\mu, \quad 1 \leq a < b \leq k.
    \end{equation}
    In our chosen (ungraded) convention, $\gamma^{(a)}_\mu$ are Majorana operators acting as $\gamma_\mu$ on the $a$-th replica and as identity elsewhere. They satisfy the commutation relations $[\gamma^{(a)}_\mu, \gamma^{(b)}_\nu] = 0$ for $a \neq b$ and $\{ \gamma^{(a)}_\mu, \gamma^{(a)}_\nu \} = 2 \delta_{\mu \nu} \mathbb{1}$ for $a = b$. As the kernels of graded and ungraded bridge operators coincide (Remark \ref{rem:graded-ungraded-bridges}), the definition of $\GSym^k(\Hcal)$ is independent of this choice.
\end{definition}

\begin{definition} \label{def:gaussian-invariant}
Let $\rho$ be a density operator on $\Hcal^{\otimes k}$. We say that $\rho$ is Gaussian-invariant if and only if it satisfies: 
\begin{equation}
[\rho, \Lambda_{ab}] = 0, \qquad 1 \leq a < b \leq k,
\end{equation} where $\Lambda_{ab}$ are as given in Definition \ref{def:gsym}. If $\supp(\rho) \subseteq \GSym^k(\Hcal)$ then $\rho$ is Gaussian-invariant, though the converse is not true in general.
\end{definition}

Our results concerning partial traces and purifications on enlarged spaces $\Hcal \otimes \Kcal$ make use of pair-parity operators $Q_{ab}$. These objects enable us to write each enlarged bridge operator $\bar{\Lambda}_{ab}$ on $(\Hcal \otimes \Kcal)^{\otimes k}$ as a sum of two bridge operators, acting independently on $\Hcal^{\otimes k}$ and its copy $\Kcal^{\otimes k}$. The pair-parity operators are defined as follows:

\begin{definition}
A pair-parity operator $Q_{ab}$ is a unitary operator acting on replicas $a$ and $b$ of $\Hcal^{\otimes k}$, defined as:
\begin{equation}
    Q_{ab} = \Gamma^{(a)} \Gamma^{(b)},
\end{equation}
where $\Gamma^{(a)} = (-i)^n \prod_{\mu = 1}^{2n} \gamma^{(a)}_\mu = \prod_{j=1}^n Z^{(a)}_j$ is the parity operator acting on replica $a$ and as identity elsewhere. 
\end{definition}
Since each $\Gamma^{(a)}$ is a product of all Majorana operators, it follows that $\{ \Gamma^{(a)}, \Lambda_{ab} \} = 0$ and $[Q_{ab}, \Lambda_{ab}] = 0$. For our exponential de Finetti theorem, we make use of the notion of ${k \choose m}$-i.i.d., or almost-i.i.d. states, which were introduced in \cite{Renner_2007}:

\begin{definition}
    Let $\nu = |\nu \rangle \langle \nu |$ be a rank-one projector on $\Hcal$. A vector $|\Psi \rangle \in \Hcal^{\otimes k}$ is called ${k \choose m}$-i.i.d. in $\nu$ if there exists a permutation $\Pi \in S_k$ such that:
    \begin{equation}
    (\nu^{\otimes m} \otimes \mathbb{1}^{\otimes k - m}) \Pi |\Psi \rangle = \Pi |\Psi \rangle
    \end{equation}
    Intuitively, this means that the state $| \Psi \rangle$ is of the form $|\nu\rangle $ on at least $m$ copies out of $k$ total replicas of $\Hcal$. 
\end{definition}

To construct approximants for the reduced state, we require projectors onto the spaces of ${k \choose k-r}$-i.i.d. states, where $0 \leq r \leq k - 1$. The parameter $r$ roughly gives the maximum number of replicas which are not identical in the approximant, i.e. if $r=0$ then the approximant is fully i.i.d., and if $r=k-1$ then it is not guaranteed to contain any two identical copies. The definition of the projectors in the replicated matchgate setting cleanly carries over from \cite{Renner_2007}:

\begin{definition} \label{def:iid}
    Let $\nu_0 = |\mathbf{0} \rangle \langle \mathbf{0} |$ be a fixed, rank-one projector on $\Hcal$. For $k, r\in \mathbb{N}$,  $0 \leq r \leq k - 1$, and any matchgate unitary $U \in \mathfrak{M}_n$, an operator $P_U^{k, r}$ is defined to be a projector onto the subspace of $\Hcal^{\otimes k}$ spanned by all ${k \choose k - r}$-i.i.d. vectors in $\nu = |\nu \rangle \langle \nu | = U \nu_0 U^\dag = U |\mathbf{0} \rangle \langle \mathbf{0} | U^\dag$. In particular,
    \begin{equation}
    P_U^{k, r} = U^{\otimes k} P_{\mathbb{1}}^{k, r} (U^\dag)^{\otimes k}.
    \end{equation}
    \noindent When $r=0$, the projector $P^{k, 0}_{\mathbb{1}} = |\mathbf{0} \rangle \langle \mathbf{0} |^{\otimes k}$.
\end{definition}

\

\paragraph{An Exponential Gaussian de Finetti Theorem.} We now present our first main result, which adapts the exponential de Finetti theorem of \cite{Renner_2007} to the Gaussian-symmetric setting. 

\begin{theorem}
\label{thm:fermionic-gaussian-definetti}
Let \(k,m\geq1\), \(0\leq r\leq k-1\) and $\mathcal{H} \cong (\mathbb{C}^2)^{\otimes n}$ for $n \geq 2$. Let
\(\rho^{k+m}\) be a density operator on \(\Hcal^{\otimes(k+m)}\)
supported on the Gaussian-symmetric subspace:
\begin{equation}
    \supp (\rho^{k+m}) \subseteq \GSym^{k+m}(\Hcal).
\end{equation}
Then, there exists a measurable family of subnormalized nonnegative
operators $\{\tilde\rho_U^k\}_{U\in\Mcal_n}$ on $\Hcal^{\otimes k}$, such that:
\begin{equation}
    \operatorname{supp} (\tilde\rho_U^k)
    \subseteq
    \operatorname{Ran} (P_U^{k,r})
\end{equation}
for every \(U\in\Mcal_n\), $\tilde{\rho}_U^k$ is at most rank one whenever $\rho^{k+m}$ is rank one, and:
\begin{equation}
    \left\|
        \operatorname{Tr}_m\left( \rho^{k+m} \right)
        -
        \int_{\Mcal_n}\tilde\rho_U^k\,dU
    \right\|_1
    \leq
    \delta.
\end{equation}
An upper bound for $\delta$ is given by:
\begin{equation}
\delta \leq 3 \sum_{t=0}^{k-r-1} \sum_{a=0}^{r+1} \binom{r+t}{r} \binom{r+1}{a}  (-1)^a \frac{D_m}{D_{m+a+t}}, \label{eq:delta-full}
\end{equation}
where $D_m := \dim (\GSym^m(\Hcal))$. Alternatively, $\delta$ admits the looser upper bounds:
\begin{equation}
\delta \leq 3 \binom{k}{r+1} \frac{(L_n)_{r+1}}{(m+1)^{r+1}} \leq \frac{3}{(r+1)!} \lambda_r^{r+1}, \label{eq:delta-simple}
\end{equation}
where $L_n := n(n-1)/2$, the subscript $(L_n)_{r+1}$ denotes the rising factorial, and:
\begin{equation}
    \lambda_r := \frac{\left( k - \frac{r}{2} \right)\left(L_n + \frac{r}{2} \right)}{\left(m+1\right)}.
\end{equation}
\end{theorem}

\begin{corollary} \label{cor:fermionic-gaussian-renner-bound}
Under the assumptions of Theorem \ref{thm:fermionic-gaussian-definetti}, the upper bound $\delta$ also satisfies:
\begin{equation} \label{eq:renner-exact-power}
\delta \leq \frac{3 D_m}{2} \left(\frac{k}{k+m} \right)^{r+1}.
\end{equation}
This leads to a Renner-style bound:
\begin{equation} \label{eq:delta-renner-style}
\delta \leq 3 e^{-\frac{m}{m+k}(r+1) + L_n \log (m+1)}.
\end{equation}
\end{corollary}

\begin{figure*}[t]
    \centering
    \includegraphics[width=1\textwidth]{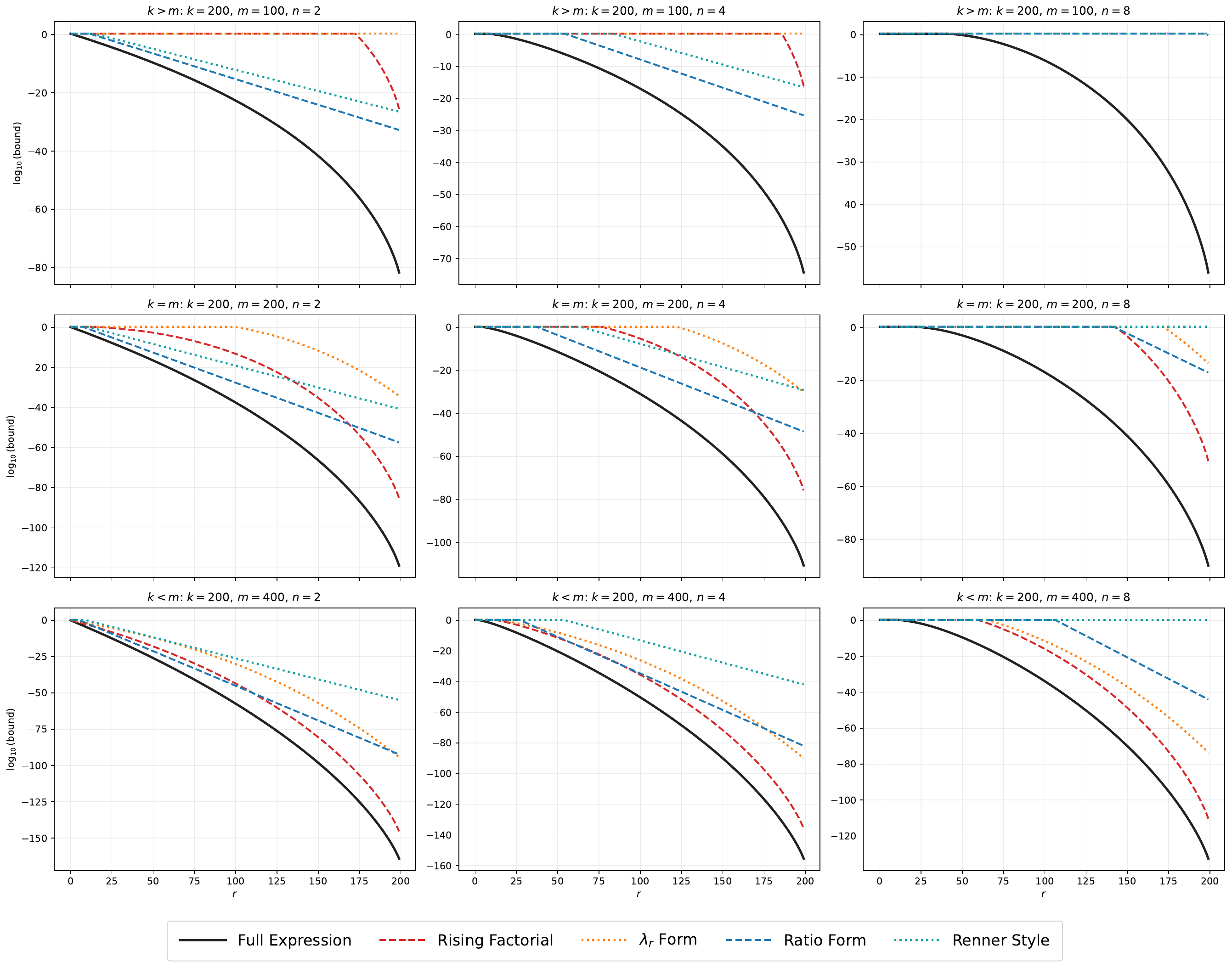}
    \caption{\textit{Numerical Comparison of Error Bounds.} Each subfigure plots upper bounds for $\delta$ in Theorem \ref{thm:fermionic-gaussian-definetti} against $r$. `Full Expression' (black line) denotes Equation \eqref{eq:delta-full}, which is the most accurate form. `Rising Factorial' and `$\lambda_r$ Form' (red and orange lines) denote the two simplified bounds in Equation \eqref{eq:delta-simple}. `Ratio Form' and `Renner Style' (blue and teal lines) denote Equations \eqref{eq:renner-exact-power} and \eqref{eq:delta-renner-style}. For both pairs, dotted lines show simplified expressions for the bounds plotted as dashed lines. Columns display results for different values of $n$  per subsystem $\mathcal{H}$, doubled between each column as prescribed by the purification in Theorem \ref{thm:gaussian-invariant}. Rows display different regimes: $k > m$ where the partial trace retains $k / (k+m) = 2/3$ of the original system, $k = m$ where $1/2$ of the system is retained, and $k < m$ where $1/3$ of the system is retained. In all cases, we observe super-exponential decay of the error $\delta$ with $r$, with the error worsening as $n$ increases.
    }
    \label{fig:definetti-illustration}
\end{figure*}

Since the operators $\tilde{\rho}^k_U$ in Theorem \ref{thm:fermionic-gaussian-definetti} are subnormalised, in practice one can replace them with density operators $\sigma^k_U = \tilde{\rho}^k_U / \Tr(\tilde{\rho}^k_U)$ with the corresponding (non-uniform) probability measure $d \mu (U) \propto \Tr(\tilde{\rho}^k_U) dU$. The approximants $\tilde{\rho}_U^k$ are almost-i.i.d. but not, in general, Gaussian-symmetric. Indeed, comparing Definitions \ref{def:gsym} and \ref{def:iid}, one sees that $\operatorname{Ran} (P_U^{k,r} ) \nsubseteq \GSym^k(\Hcal)$ for $r \geq 1$: a vector $| \Psi \rangle$ carrying $(U|\mathbf{0}\rangle)^{\otimes(k-r)}$ on $k-r$ copies is annihilated by the bridge operators $\Lambda_{ab}$ acting within those copies, but the remaining $r$ copies are unconstrained, so the bridge operators involving them are not guaranteed to annihilate $| \Psi \rangle$. If $\rho^{k+m}$ is pure, then $\tilde{\rho}_U^k$ is rank one, and the corresponding $\sigma_U^k$ is a superposition of such $| \Psi \rangle$ vectors. Each contribution is i.i.d. Gaussian on its own $k-r$ copies, but those copies may be chosen in $\binom{k}{k-r}$ ways, varying between each term in the superposition. The physical interpretation of our result is that any $k$-replica part $\Tr_m(\rho^{k + m})$ of a larger, globally Gaussian-symmetric state $\rho^{k+m}$ can be approximated by a probabilistic mixture of states $\sigma_U^k$, each of which \textit{need not be} globally Gaussian-symmetric, yet is composed of at least $k-r$ i.i.d. Gaussian states. In other words, subsystems of a globally Gaussian-symmetric system are well-approximated by statistical mixtures of $\sigma_U^k$, which are superpositions of Gaussian-symmetric states on smaller systems (when restricted to each choice of $k-r$ copies where the i.i.d. Gaussian states are supported). 

Our bounds (in particular Equations \eqref{eq:delta-simple} and \eqref{eq:delta-renner-style}) give insight into the validity of the approximation in different regimes. By Equation \eqref{eq:delta-simple}, in the region where $m \gg k$, our upper bound is manifestly super-exponentially decaying with $r$. This corresponds to the retained subsystem being a very small part of the overall system. Furthermore, for fixed $(n, k, r)$ and $m \rightarrow \infty$, we obtain a power-law decay $\delta \sim m^{-r - 1}$. On the other hand, the form in Equation \eqref{eq:delta-renner-style} is more useful in the regime of $k \gtrsim m$, where $\delta$ can be seen to also decay exponentially with $r$. This corresponds to situations where the retained system is only minimally smaller than the overall system. Numerically, using the form of Equation \eqref{eq:delta-full} we find $\delta$ to also decay super-exponentially in $r$ in situations where the two simplified expressions fail. When $r = 0$, the approximant states become fully i.i.d. and Equation \eqref{eq:delta-full} reproduces the result of \cite{Sierant_2026} up to a constant factor. 

In relation to previous work, a free-fermionic de Finetti theorem was first shown in \cite{Vershynina_2014}, in the context of approximating reduced states $\Tr_{m}(\rho^{m+1})$, where $\supp (\rho^{m+1}) \subseteq \GSym^{m+1}(\Hcal)$. This was achieved with an error bound $\delta \leq \mathcal{O}(n 2^{4n} m^{-1/3})$. For partial traces $\Tr_{m}(\rho^{k+m})$ of Gaussian-symmetric states $\supp (\rho^{k+m}) \subseteq \GSym^{k+m}(\Hcal)$, the authors of \cite{Sierant_2026} gave an exponentially improved error bound $\delta \leq \mathcal{O}(k n^2 / (k+m+1))$. Our work extends the two scenarios to approximation by $\binom{k}{k-r}$-i.i.d. states on $\Hcal^{\otimes k}$ of the form $(U | \mathbf{0} \rangle)^{\otimes k-r}$ on at least $k-r$ replicas of $\Hcal$. This offers a further exponential improvement in $\delta$ for $r > 0$. 

Another direct point of comparison is \cite{Renner_2007}, where using almost-i.i.d. states to approximate symmetric states was shown to result in an error bound of:
\begin{equation}
\delta \leq 3e^{-\frac{m}{m+k}(r+1) + d \log m}. \label{eq:renner-bound}
\end{equation}
Gaussian symmetry exponentially improves this scaling by replacing the $d \log m$ term with $L_n \log (m+1)$, which is polylogarithmic in $d$. This is a consequence of the much tighter dimensionality of the Gaussian-symmetric subspace, which forms the trivial replica $SO(k)$ sector. As a result, one can continue to obtain good convergence in the approximations for large $n$, so long as $k, m$ grow polynomially with the local qubit number. The other contrast with \cite{Renner_2007} is the regime of validity; Equation \eqref{eq:renner-bound} justifies a global representation theorem, wherein a large subsystem of an overall symmetric system is well-approximated (with exponentially decaying error in $r$) by a combination of almost-i.i.d. states. In our case, combining Equations \eqref{eq:delta-simple} and \eqref{eq:delta-renner-style} implies that this remains valid for approximating both large and small subsystems of a globally Gaussian-symmetric state. Therefore, our results show that the representation theorem of \cite{Renner_2007} not only applies to free-fermionic systems, but does so with a wider regime of validity and more favourable convergence between the approximants and the partially-traced state. For problems which are natively free-fermionic, where the original state $\rho^{k+m}$ resides in $\GSym^{k+m}(\Hcal)$ (or can be made so by a matchgate-covariant symmetrisation), we expect the applications of our result to largely mirror those of \cite{Renner_2007}, with the added benefit of tighter convergence. 

For continuous-variable systems, bosonic and fermionic Gaussian de Finetti theorems were proven in \cite{Leverrier_2009, Krumnow_2017, Leverrier_2017, Leverrier_2018}. Our theorem is complementary to these results, restricting to systems with a finite number of modes described by matchgate unitaries. In particular, \cite{Leverrier_2018} notes the absence of exponential variants in the continuous-variable setting and remarks on the potential utility of such results for quantum cryptography. It would therefore be interesting to see if Theorem \ref{thm:fermionic-gaussian-definetti} can be applied to obtain discrete-variable, free-fermionic analogues of such protocols. Finally, \cite{Gross_2021} considered de Finetti theorems for the stabilizer states of the Clifford group. An interesting direction for future work would be to combine the two, and verify whether restrictions to the Clifford-matchgate group are admissible and provide any further improvements in the error bounds. Extending standard techniques to this case would require a closed-form expression for the dimension of its trivial irrep, as well as verification of whether it occurs with multiplicity one. 

\

\paragraph{Characterising Gaussian-Invariant States.} Our second main result characterises the family of states obtained by partially tracing Gaussian-symmetric states. It is well-known that in the original setting, de Finetti theorems apply to any partial trace of a permutation-symmetric state, giving a wider family of permutation-invariant states which admit an i.i.d. approximation. However, in the replicated matchgate case, the question of partial traces has previously gone unanswered. One of the major contributions of recent work on matchgate commutants \cite{Sierant_2026, braccia_2026, lastres_2026} is the extension of the result that any Gaussian state $\rho$ on $\Hcal$ satisfies $\Lambda (\rho \otimes \rho) = 0$ on $\Hcal \otimes \Hcal$ \cite{bravyi_2004} to the multi-replica setting, where the bridge operators $\Lambda_{ab}$ annihilate any Gaussian-symmetric state on $\Hcal^{\otimes k}$. Similarly to how impure Gaussian states on $\Hcal$ are known to satisfy $[\Lambda, \rho \otimes \rho]=0$, we find that any Gaussian-invariant state $\rho$ on $\Hcal^{\otimes k}$ commutes with all bridge operators. Furthermore, these states come from partial traces of Gaussian-symmetric states of a larger Hilbert space, and always admit a Gaussian-symmetric purification. Most importantly, our purification theorem enables the application of Theorem \ref{thm:fermionic-gaussian-definetti}, as well as its fully-i.i.d. variant \cite{Sierant_2026}, to any Gaussian-invariant state.

\begin{theorem} \label{thm:gaussian-invariant}
For a density operator $\rho$ on $\Hcal^{\otimes k}$, the following are equivalent:
\begin{enumerate}
    \item $\rho$ is Gaussian-invariant.
    \item $\rho = \Tr_{\Kcal^{\otimes k}}(\bar{\rho})$ for some density operator $\bar{\rho}$ supported on $\GSym^k(\Hcal \otimes \Kcal)$, where $\Kcal \cong \Hcal$.
    \item $\rho$ admits a purification $|\bar{\rho} \rangle \in \GSym^k(\Hcal \otimes \Kcal)$.
\end{enumerate} 
\end{theorem}
For the analogous statement in the graded convention, see Remark \ref{rem:graded-ungraded-theorems} in the Apendix. Lemma \ref{lem:parity-permutation-commutation} additionally shows that Gaussian-invariant states commute with all pair-parity operators $Q_{ab}$ and permutations $\Pi \in S_k$. The latter implies they are a subset of the permutation-invariant states.  Finally, a related result may be found in \cite{Melo_2013}. For single-copy convex combinations of Gaussian states $\rho$, the authors considered the problem of finding Gaussian-symmetric extensions, i.e. states $\rho_\mathrm{ext}$ with support on $\GSym^{k+1}(\Hcal)$ such that $\Tr_k (\rho_\mathrm{ext}) = \rho$. In this work we do not consider Gaussian-symmetric extensions, but rather purifications of $k$-replica states $\rho$ on the locally-enlarged $\GSym^k(\Hcal \otimes \Kcal)$, which always exist.

\

\paragraph{Examples of Gaussian-invariant States.} Theorem \ref{thm:gaussian-invariant} shows that the partial trace over $\Kcal^{\otimes k}$ of any Gaussian-symmetric state $\bar{\rho}$ supported on $\GSym^k(\Hcal \otimes \Kcal)$ is Gaussian-invariant. Similarly to partial traces over $\Sym^k(\Hcal \otimes \Kcal)$, if the state $\bar{\rho}$ is a simple tensor product then $\Tr_{\Kcal^{\otimes k}}(\bar{\rho})$ yields another Gaussian-symmetric state. To see this, take a tensor product $\psi^{\otimes k}$, such that $\psi$ is a pure Gaussian state on $\Hcal$. Then $\Lambda_{ab} \psi^{\otimes k} = 0$, so $\psi^{\otimes k} \in \GSym^k(\mathcal{H})$. On the larger space $(\Hcal \otimes \Kcal)^{\otimes k} \cong \Hcal^{\otimes k} \otimes \Kcal^{\otimes k}$, the bridge operators $\bar{\Lambda}_{ab}$ can be shown to take on the form: 
\begin{equation}
\bar{\Lambda}_{ab} = \Lambda^\Hcal_{ab} \otimes \mathbb{1}^\Kcal_{ab} + Q^\Hcal_{ab} \otimes \Lambda^\Kcal_{ab},
\end{equation}
where $Q^\Hcal_{ab}$ is a pair-parity operator acting on the $a$-th and $b$-th replicas of $\Hcal^{\otimes k}$. The corresponding state $\bar{\psi} = \psi^{\otimes 2k}$ clearly satisfies $\Tr_{\Kcal^{\otimes k}}(\bar{\psi}) = \psi^{\otimes k}$. Furthermore, it is Gaussian-symmetric, since:
\begin{equation}
    \bar{\Lambda}_{ab} \bar{\psi} = \Lambda^\Hcal_{ab} \psi^{\otimes k} \otimes \psi^{\otimes k} + Q^\Hcal_{ab} \psi^{\otimes k} \otimes \Lambda^\Kcal_{ab} \psi^{\otimes k} = 0.
\end{equation}
It also follows that a purification of $\psi^{\otimes k}$ on $\GSym^k(\Hcal \otimes \Kcal)$ is of the form $\psi^{\otimes 2k}$. This suggests that to yield non-Gaussian-symmetric purifications, one must consider states $\bar{\rho}$ on $\GSym^k(\Hcal \otimes \Kcal)$ that contain genuine entanglement between the subsystems $\Hcal$ and $\Kcal$. As a minimal example with $n=1$ and $k=2$, we may take $\bar{\rho} = | \Omega \rangle \langle \Omega |^{\otimes 2}_{\Hcal \Kcal}$ where $| \Omega \rangle_{\Hcal \Kcal} = \frac{1}{\sqrt{2}}(|00\rangle + |11\rangle)$ is a maximally entangled state on $\Hcal \otimes \Kcal$. In this case, the single bridge operator (ordering the subsystems it acts on as $\Hcal \otimes \Kcal \otimes \Hcal \otimes \Kcal$) is given by:
\begin{equation}
\bar{\Lambda}_{12} = X_1 X_3 + Y_1 Y_3 + Z_1 X_2 Z_3 X_4 + Z_1 Y_2 Z_3 Y_4,
\end{equation}
and it is straightforward to verify that $\bar{\Lambda}_{12} |\Omega\rangle^{\otimes 2} = 0$. On the other hand, the partial trace over $\Kcal^{\otimes 2}$ yields a maximally mixed state: 
\begin{equation}
    \Tr_{\Kcal^{\otimes 2}}(\bar{\rho}) = \frac{1}{4} \mathbb{1}_{\Hcal^{\otimes 2}},
\end{equation}
which trivially commutes with $\Lambda_{12}$, $Q_{12}$ and any permutation $\Pi \in S_2$, but is not Gaussian-symmetric.

The above example demonstrates an instance of a Gaussian-invariant (but not Gaussian-symmetric) state which is mixed. In fact, one can easily show that for $k \geq 3$ replicas any pure Gaussian-invariant state must be Gaussian-symmetric. To see this, assume that $|\Psi \rangle$ is a pure Gaussian-invariant state. It follows that $\Lambda_{ab} |\Psi \rangle \langle \Psi | = |\Psi \rangle \langle \Psi | \Lambda_{ab}$ for all $a<b$. Hence, $|\Psi \rangle$ is an eigenvector of all bridge operators $\Lambda_{ab}$:
\begin{equation}
     \Lambda_{ab} |\Psi \rangle = \Lambda_{ab} |\Psi \rangle \langle \Psi | \Psi \rangle = \langle \Psi | \Lambda_{ab} | \Psi \rangle |\Psi \rangle = \lambda_{ab} |\Psi \rangle. 
\end{equation}
On the other hand, by virtue of the anticommutation relations $\{\gamma^{(b)}_\mu, \gamma^{(b)}_\nu\} = 2 \delta_{\mu\nu}$, the bridge operators $\Lambda_{ab}$ satisfy $\{\Lambda_{ab}, \Lambda_{bc}\} = 2 \Lambda_{ac}$. Applying $|\Psi \rangle$ to both sides gives:
\begin{equation}
    (\lambda_{ab} \lambda_{bc} + \lambda_{bc} \lambda_{ab}) |\Psi \rangle = 2 \lambda_{ac} |\Psi \rangle.
\end{equation}
Repeating this cyclically then gives:
\begin{align}
\lambda_{ac} &= \lambda_{ab} \lambda_{bc}, & \lambda_{ab} &= \lambda_{ac} \lambda_{bc}, & \lambda_{bc} &= \lambda_{ab} \lambda_{ac}.
\end{align}
If one of the $\lambda_{ab}$ is zero, then all three must vanish. Otherwise, this implies that $\lambda^2_{ab} = \lambda^2_{ac} = \lambda^2_{bc} = 1$. However, this is impossible, since $\operatorname{spec}(\Lambda_{ab}) = \{2r - 2n : r = 0, \dots, 2n \}$, so all three must vanish. For $k \geq 3$ every pair $(a, c)$ admits a third index, which implies that $\lambda_{ab} = 0$ for all $a<b$. Hence, $|\Psi \rangle$ is annihilated by all bridge operators, and is Gaussian-symmetric. On the other hand, for $k=2$ the argument does not hold. As a simple example, consider $n=1$ with the singlet state $|S\rangle = (|01\rangle - |10\rangle) / \sqrt{2}$ and $\Lambda = X \otimes X + Y \otimes Y$. This state is pure, its outer product commutes with $\Lambda$, yet $\Lambda |S \rangle = -2|S\rangle \neq 0$. 

Mixed, Gaussian-invariant states outside of $\GSym^k(\Hcal)$ provide a rich family of interesting and physically relevant objects. For example, consider states of the form $\sigma^{\otimes k}$ where $\sigma$ is a mixed Gaussian state on $\Hcal$ (for instance, $\sigma$ could be a fermionic Gibbs state \cite{ramkumar_2026}). It is well-known that on $\Hcal^{\otimes 2}$ such states satisfy $[\Lambda, \sigma \otimes \sigma] = 0$ with $\Lambda(\sigma \otimes \sigma) \neq 0$ \cite{bravyi_2004}. For many-replica systems, it follows straightforwardly that $[\Lambda_{ab}, \sigma^{\otimes k}] = 0$ and $\Lambda_{ab} \sigma^{\otimes k} \neq 0$. Finally, non-i.i.d. Gaussian-invariant states also exist. By the dual pair decomposition $\mathcal{H}^{\otimes k}$ under the tensor action of $\mathfrak{M}_n$ and its commutant:
\begin{equation}
\Hcal^{\otimes k} \cong \bigoplus_{\nu} W_\nu^{\Mcal_n} \otimes V_\nu^{SO(k)},
\end{equation}
it is clear that states $\rho$ of the form:
\begin{equation}
    \rho=\bigoplus_\nu p_\nu\,
\rho_\nu^{W} \otimes
\frac{\mathbb 1_{V_\nu}}{\dim V_\nu}
\end{equation}
can instantiate non-i.i.d. Gaussian-invariant states. Our purification theorem enables approximation of such states by convex combinations of i.i.d., or almost-i.i.d. states, with the purification process leading to a tractable overhead in the error bound.

\

\paragraph{De Finetti Theorems for Gaussian-invariant States.} Theorem \ref{thm:gaussian-invariant} allows us to apply Theorem \ref{thm:fermionic-gaussian-definetti} to any Gaussian-invariant state $\rho$ on $\Hcal^{\otimes k}$ via purification into a Gaussian-symmetric state $\bar{\rho}$ on $(\Hcal \otimes \Kcal)^{\otimes k}$. This process essentially replaces each subsystem $\Hcal$ with its two-fold enlargement $\Hcal \otimes \Kcal$, squaring the local dimension and doubling the local qubit number. Remarkably, both the bound in Theorem \ref{thm:fermionic-gaussian-definetti} and the result of \cite{Sierant_2026} scale with respect to the latter, meaning that the purification process minimally affects the error bound. In contrast, the finite quantum de Finetti theorem for symmetric states given in \cite{Christandl_2007} gives a bound $\delta = \mathcal{O}(d)$, with the purification process squaring the local dimension. For many-qubit systems, this would suggest a significant multiplicative increase ($d^2 / d = 2^n$) in the error bound when applied to Gaussian-invariant states -- something that is avoided entirely, owing to the higher degree of structure and lower dimension of the Gaussian-symmetric subspace.

\begin{corollary}
Let $\rho^{k+m} \in \mathcal{H}^{\otimes k+m}$ be Gaussian-invariant, i.e. $[\rho^{k+m}, \Lambda_{ab}] = 0$ for all $a < b$. By the results of Section IV G in
\cite{Sierant_2026} and Theorem \ref{thm:gaussian-invariant}, there exists a pure
Gaussian-symmetric state $\psi \in (\mathcal{H} \otimes \mathcal{K})^{\otimes {k+m}}$
with $\mathcal{K} \cong \mathcal{H}$ such that
$\rho^{k+m} = \Tr_{\Kcal^{\otimes {k+m}}}(\psi)$. Let $\mu$ be the probability measure on pure
Gaussian states $\phi$ on $\mathcal{H} \otimes \mathcal{K}$ given by
$d\mu(\phi) = D_{k+m} \Tr(\phi^{\otimes {k+m}}\psi)\,d\phi$, and write
$\sigma_\phi = \Tr_{\Kcal}(\phi)$ for the induced Gaussian-invariant state on $\mathcal{H}$.
Then, the reduced state of $\rho^{k+m}$ on any $k$ copies satisfies:
\begin{equation}
    \left\| \Tr_{m}(\rho^{k+m}) - \int d\mu(\phi)\, \sigma_\phi^{\otimes k} \right\|_1
    \;\leq\; \frac{4 k n(2n - 1)}{k + m + 1}.
\end{equation}
\end{corollary}

\begin{proof}
Applying the Gaussian de Finetti theorem (Equation (188) of \cite{Sierant_2026} or Theorem \ref{thm:fermionic-gaussian-definetti} for $r=0$) to $\psi$, whose subsystems $\mathcal{H} \otimes \mathcal{K}$ have $2n$ qubits, gives (in trace norm):
\begin{align}
    \left\| \Tr_{m}(\psi) - \int d\mu(\phi)\, \phi^{\otimes k} \right\|_1
    \leq \frac{2 k (2n)(2n-1)}{k + m +1}.
\end{align}
Now, apply the partial trace $\Tr_{\Kcal^{\otimes k}}$ to both arguments. It maps
$\Tr_{m}(\psi) \rightarrow \Tr_{m}(\rho^{k+m})$ and
$\phi^{\otimes k} \rightarrow \sigma_\phi^{\otimes k}$. By monotonicity of the trace norm under CPTP maps, the bound is preserved.
\end{proof}

Although the induced state $\sigma_\phi$ is Gaussian-invariant, it is not necessarily Gaussian-symmetric or pure. However, since $\sigma_\phi$ is a partial trace of a pure Gaussian state $\phi$, in practice it is always sufficient to prepare a pure Gaussian state $\phi$ on each $\Hcal \otimes \Kcal$ subsystem and discard its $\Kcal$ part. In practice, one can approximate any Gaussian-invariant state $\rho$ via its purification $\psi$ as follows:
\begin{enumerate}
    \item Draw $\phi \sim d \mu(\phi) = D_{k+m} \Tr(\phi^{\otimes (k + m)} \psi) d \phi$.
    \item Form $\phi^{\otimes k}$ on $(\Hcal \otimes \Kcal)^{\otimes k}$ and discard the $\Kcal^{\otimes k}$ part to obtain $\sigma_\phi^{\otimes k}$ on $\Hcal^{\otimes k}$. 
    \item Repeat the above steps many times to obtain a convex combination of $\sigma_\phi^{\otimes k}$, which approximates $\Tr_{m}(\rho^{k+m})$ with an error $\delta \leq \mathcal{O}(kn^2/(k+m+1))$.
\end{enumerate}

Similarly, Theorem \ref{thm:fermionic-gaussian-definetti} can be applied to any Gaussian-invariant state $\rho^{k + m}$ on $\Hcal^{\otimes (k+m)}$, via purification into a Gaussian-symmetric state $\bar{\rho}^{k + m}$ on $(\Hcal \otimes \Kcal)^{\otimes (k+m)}$, with the approximants $\sigma_U^k$ in the theorem replaced by $\Tr_{\Kcal^{\otimes k}}(\bar{\sigma}_U^k)$. By the monotonicity of the trace norm under CPTP maps, the error bound is preserved, with the coefficient $L_n$ replaced by $L_{2n} \sim 4 L_n$ to account for the doubling of the local qubit number. This produces a multiplicative increase in the error bound of order:
\begin{equation}
\mathcal O\!\left((m+1)^{L_{2n}-L_n}\right) =(m+1)^{(3n^2-n)/2} =m^{\Theta(n^2)},
\end{equation} where for comparison, the bound in Equation \ref{eq:renner-bound} given in \cite{Renner_2007} would produce a multiplicative increase of order:
\begin{equation}
\mathcal{O}((m)^{d^2 - d}) = m^{\Theta(2^{2n})}.
\end{equation} 

\
\paragraph{Application to fermionic Gaussian optimization.}
Fermionic Gaussian states provide a variational description of fermionic
systems, including mean-field states and states with pairing
correlations~\cite{BachLiebSolovej1994,BardeenCooperSchrieffer1957}.
For a Hamiltonian (given as a sum of terms of bounded degree in Majorana
operators), its expectation with respect to a Gaussian state can be evaluated from the covariance matrix using Wick's theorem. Optimizing this expectation over
the Gaussian family is a nonconvex problem, and a local variational
calculation is not guaranteed find the global minimum.

Let $\nu_U=U|\mathbf{0}\rangle\langle\mathbf{0}|U^\dagger$, with
$U\in\Mcal_n$. We use the full matchgate orbit of the vacuum, including
both parity sectors. The optimal Gaussian variational energy is represented as
\begin{equation}
    E_{\mathrm G}(H)
    :=
    \inf_{U\in\Mcal_n}\Tr(H\nu_U).
\end{equation}
Any Gaussian trial state gives an upper bound on this energy and lower
bounds can be obtained by optimizing over Gaussian-symmetric extensions.
Such extensions have previously been used to approximate the set of
convex combinations of pure Gaussian states~\cite{Vershynina_2014}.
Here we use the Gaussian de Finetti bounds to quantify the error in
the resulting optimization.

More generally, for a Hermitian operator $C_t$ on $t$ replicas, define
\begin{equation}
\label{eq:main-gaussian-optimization}
    \omega_t(C_t)
    :=
    \sup_{U\in\Mcal_n}
    \Tr\left[C_t\nu_U^{\otimes t}\right].
\end{equation}
For $t=1$, this includes energies through
$\omega_1(-H)=-E_{\mathrm G}(H)$ and squared overlaps through
$C_1=|\psi\rangle\langle\psi|$. 
The parameter $t$ above counts replicas (rather than the degree of the
objective in covariance matrix entries). 
At extension level $\ell\geq t$, consider
\begin{align}
\label{eq:main-gaussian-extension}
    \omega_\ell^{\mathrm{ext}}(C_t)
    :=
    \sup_{\rho^\ell}\;&
    \Tr\left[
        \left(C_t\otimes\mathbb{1}^{\otimes(\ell-t)}\right)
        \rho^\ell
    \right],
    \\
    \text{subject to }\;&
    \rho^\ell\geq0,\quad \Tr(\rho^\ell)=1,
    \nonumber\\
    &\supp(\rho^\ell)\subseteq\GSym^\ell(\Hcal).
    \nonumber
\end{align}
Every Gaussian tensor power $\nu_U^{\otimes\ell}$ is feasible.
Moreover, tracing out one replica preserves feasibility at the
preceding level. Hence
\begin{equation}
\label{eq:main-gaussian-hierarchy}
    \omega_t(C_t)
    \leq
    \omega_{\ell+1}^{\mathrm{ext}}(C_t)
    \leq
    \omega_\ell^{\mathrm{ext}}(C_t).
\end{equation}
For energies, set
$E_\ell^{\mathrm{ext}}:=-\omega_\ell^{\mathrm{ext}}(-H)$.
Any Gaussian trial state $\nu_{\mathrm{trial}}$ then gives
\begin{equation}
\label{eq:main-gaussian-trial-bounds}
    E_\ell^{\mathrm{ext}}
    \leq
    E_{\mathrm G}(H)
    \leq
    \Tr(H\nu_{\mathrm{trial}}).
\end{equation}
The difference between the trial energy and $E_\ell^{\mathrm{ext}}$
bounds the `suboptimality' of that trial state. Using the above results, we can now quantify this.

\begin{proposition}
\label{prop:gaussian-optimization-gap}
Let $n\geq2$, $t\geq1$ and $\ell>t$ be integers. Then
\begin{align}
\label{eq:main-gaussian-gap}
    0
    &\leq
    \omega_\ell^{\mathrm{ext}}(C_t)-\omega_t(C_t)
    \nonumber\\
    &\leq
    3\,\operatorname{osc}(C_t)
    \left(1-\frac{D_{\ell-t}}{D_\ell}\right)
    \nonumber\\
    &\leq
    \frac{3tL_n}{\ell+1}\operatorname{osc}(C_t),
\end{align}
where $D_j=\dim(\GSym^j(\Hcal))$, $L_n=n(n-1)/2$, and
$\operatorname{osc}(C_t):=\lambda_{\max}(C_t)-\lambda_{\min}(C_t)$.
\end{proposition}

The bound follows by applying the fully i.i.d. estimate to the
$t$ replicas which feature in the objective. It gives convergence of the
hierarchy for fixed $n$ and $t$. The proof also gives a Gaussian state
$\nu_{\mathrm{round}}$ such that
\begin{equation}
\label{eq:main-gaussian-energy-certificate}
    \Tr(H\nu_{\mathrm{round}})
    \leq
    E_\ell^{\mathrm{ext}}
    +
    \frac{3L_n}{\ell+1}\operatorname{osc}(H).
\end{equation}
It is worth noting that this is an existence statement. The upper bound in Equation~\eqref{eq:main-gaussian-trial-bounds} can instead use any
Gaussian state obtained by a separate variational calculation.
These bounds are valid for optimization within the Gaussian family and do
not bound its error relative to the exact interacting ground-state
energy.

Without additional constraints, the optimization at each level reduces
to finding the largest eigenvalue of a $D_\ell\times D_\ell$ matrix,
or the smallest eigenvalue for energy minimization. For a fixed number
of fermionic modes $n$, the matrix dimension grows polynomially with
the number of replicas $\ell$. However, the exponent is $n(n-1)/2$,
so increasing the number of modes can substantially increase the
computational cost. In the Appendix, we prove
the bound and derive a further estimate from almost-i.i.d. states.
We also discuss the effect of additional constraints and the
extension to optimization over mixed Gaussian states.

\

\textit{Acknowledgements.} J.B. gratefully acknowledges the hospitality of the ICTQT and the University of Gdańsk, where the majority of this work was carried out. S.S. acknowledges support from the Royal Society University Research
Fellowship. M.S. acknowledges support by the IRA Programme, project no. FENG.02.01-IP.05-0006/23, financed by the
FENG program 2021-2027, Priority FENG.02, Measure FENG.02.01, with the support of the FNP. 

\clearpage
\onecolumngrid
\begin{center}
  \textbf{\large Appendix: Exponential de Finetti Theorems for Fermionic Gaussian States}
\end{center}

\setcounter{equation}{0}
\setcounter{figure}{0}
\setcounter{table}{0}
\setcounter{page}{1}
\renewcommand{\theequation}{S\arabic{equation}}
\setcounter{secnumdepth}{3}
\setcounter{section}{0}
\renewcommand{\thesection}{\Roman{section}}

\section{Fermionic Gaussian optimization}
\label{app:gaussian-optimization}

The optimization at level $\ell$ uses a state on $\ell$ replicas,
but the objective depends on only $t$ of them. We show that the
reduced state on these $t$ replicas is close to a mixture of states
$\nu_U^{\otimes t}$. The expectation of the objective in this
mixture is an average of its values on Gaussian states. At least
one Gaussian state must have a value as large as this average.
This gives the bound in
Proposition~\ref{prop:gaussian-optimization-gap}.

\paragraph{Expressing the objective using replicas.}
Recall that
$\nu_U=U|\mathbf{0}\rangle\langle\mathbf{0}|U^\dagger$,
with $U\in\Mcal_n$. An energy expectation $\Tr(H\nu_U)$ requires
one replica. A product of two expectation values can be written
using two replicas:
\begin{equation}
    \Tr(A\nu_U)\Tr(B\nu_U)
    =
    \Tr\left[(A\otimes B)\nu_U^{\otimes2}\right].
\end{equation}
More generally, any real polynomial of degree at most $t$ in the
matrix entries of $\nu_U$ can be written as
\begin{equation}
\label{eq:app-polynomial-lift}
    p(\nu_U)
    =
    \Tr\left[C_t\nu_U^{\otimes t}\right]
\end{equation}
for a suitable Hermitian operator $C_t$. A term of degree $s<t$
can be multiplied by $(\Tr\nu_U)^{t-s}=1$ to put it in this form.
The number $t$ is the degree in the density matrix entries.
It need not equal the degree in covariance matrix entries.
In particular, an interacting energy still requires only one
replica, because it is linear in $\nu_U$.

\paragraph{Why increasing the number of replicas improves the bound.}
Every vector in $\GSym^\ell(\Hcal)$ is unchanged by permutations
of the replicas, since this space is spanned by vectors
$(U|\mathbf{0}\rangle)^{\otimes\ell}$.
A state supported on this space has the same reduced state on
any choice of $t$ replicas, after identifying their tensor factors.
The value of the objective is thus independent of which $t$
replicas are used. Only the total number $\ell$ matters in the
optimization. A choice of $k$ and $m$ with $k+m=\ell$ is needed
only when applying the de Finetti theorem.

The same spanning vectors give
\begin{equation}
    \GSym^{\ell+1}(\Hcal)
    \subseteq
    \GSym^\ell(\Hcal)\otimes\Hcal.
\end{equation}
Thus, if a state satisfies the constraints at level $\ell+1$,
tracing out its last replica gives a state satisfying the
constraints at level $\ell$. Its reduced state on the first
$t$ replicas is unchanged. The maximum at level $\ell+1$ cannot
exceed the maximum at level $\ell$.
Also, every Gaussian tensor power $\nu_U^{\otimes\ell}$ satisfies
the constraints at level $\ell$. These two observations prove
Equation~\eqref{eq:main-gaussian-hierarchy}.

\paragraph{Making the weights sum to one.}
The de Finetti theorem gives positive operators whose total trace
can be less than one. To obtain a density operator, we add the
remaining weight to a fixed state.

Suppose that $\sigma$ is a density operator and $\tau\geq0$
satisfies
$p:=\Tr\tau\leq1$ and $\|\sigma-\tau\|_1\leq\delta$.
Taking the trace gives $0\leq1-p\leq\delta$.
For any density operator $\zeta$, the operator
$\widehat{\tau}:=\tau+(1-p)\zeta$ has trace one, and
\begin{equation}
\label{eq:app-normalized-error}
    \frac12\|\sigma-\widehat{\tau}\|_1
    \leq
    \frac12\left(\delta+1-p\right)
    \leq
    \delta.
\end{equation}
This construction also works when $p=0$.

To see why $p\leq1$ holds here, recall the construction in
Equations~\eqref{def:rhok} and~\eqref{def:rhokbar}.
The positive operators $\rho_U^k$ satisfy
$\int\Tr(\rho_U^k)\,dU=1$, and the theorem uses
$\tilde{\rho}_U^k=P_U^{k,r}\rho_U^kP_U^{k,r}$.
Since projection cannot increase the trace of a positive
operator, $\int\Tr(\tilde{\rho}_U^k)\,dU\leq1$.

\begin{proof}[Proof of Proposition~\ref{prop:gaussian-optimization-gap}]
Let $\rho^\ell$ satisfy the constraints in
Equation~\eqref{eq:main-gaussian-extension}, and set
$\sigma^t:=\Tr_{\ell-t}(\rho^\ell)$.
We only need to approximate this reduced state.
Apply Theorem~\ref{thm:fermionic-gaussian-definetti} with
$k=t$, $m=\ell-t$ and $r=0$.
For $r=0$, every operator in the resulting integral is a
nonnegative multiple of $\nu_U^{\otimes t}$.
The intermediate terms in the sum in
Equation~\eqref{eq:delta-full} cancel, giving
\begin{equation}
    \left\|
        \sigma^t-\int w(U)\nu_U^{\otimes t}\,dU
    \right\|_1
    \leq
    \delta,
    \qquad
    \delta:=3\left(1-\frac{D_{\ell-t}}{D_\ell}\right),
\end{equation}
where $w(U)\geq0$ and $p:=\int w(U)\,dU\leq1$.
Add the remaining weight $1-p$ to $\nu_0^{\otimes t}$, where
$\nu_0=|\mathbf{0}\rangle\langle\mathbf{0}|$.
The weights now define a probability measure $\mu$, and
Equation~\eqref{eq:app-normalized-error} gives
\begin{equation}
\label{eq:app-iid-normalized-error}
    \frac12\left\|
        \sigma^t-\int\nu_U^{\otimes t}\,d\mu(U)
    \right\|_1
    \leq
    \delta.
\end{equation}
For two density operators, a trace distance of at most $\delta$
changes the expectation of $C_t$ by at most
$\operatorname{osc}(C_t)\delta$.
Here $\operatorname{osc}(C_t)$ is the difference between the
largest and smallest eigenvalues of $C_t$.
It follows that
\begin{align}
    \Tr(C_t\sigma^t)
    &\leq
    \int\Tr\left[C_t\nu_U^{\otimes t}\right]d\mu(U)
    +
    \operatorname{osc}(C_t)\delta
    \nonumber\\
    &\leq
    \omega_t(C_t)+\operatorname{osc}(C_t)\delta.
\end{align}
The second inequality holds because an average cannot exceed
the largest value being averaged.
Maximizing over $\rho^\ell$ proves the first upper bound in
Equation~\eqref{eq:main-gaussian-gap}.

For the second bound, the dimension formula gives
\begin{align}
    1-\frac{D_{\ell-t}}{D_\ell}
    &=
    1-\prod_{1\leq i<j\leq n}
    \left(1-\frac{t}{\ell+2n-i-j}\right)
    \nonumber\\
    &\leq
    \sum_{1\leq i<j\leq n}
    \frac{t}{\ell+2n-i-j}
    \leq
    \frac{tL_n}{\ell+1}.
\end{align}
We used $1-\prod_j(1-x_j)\leq\sum_jx_j$ for $0\leq x_j\leq1$.
There are $L_n=n(n-1)/2$ terms, and every denominator is at
least $\ell+1$.
Finally, the lower bound follows because every Gaussian tensor
power is allowed in the optimization.
\end{proof}

The proof also explains why a Gaussian state can have an
objective value close to the relaxed maximum.
Choose $\rho^\ell$ attaining $\omega_\ell^{\mathrm{ext}}(C_t)$.
The average over $\mu$ is within the error above of this maximum,
so at least one Gaussian state satisfies
\begin{equation}
\label{eq:app-existential-rounding}
    \Tr\left[C_t\nu_U^{\otimes t}\right]
    \geq
    \omega_\ell^{\mathrm{ext}}(C_t)
    -
    \frac{3tL_n}{\ell+1}\operatorname{osc}(C_t).
\end{equation}
Taking $C_1=-H$ gives
Equation~\eqref{eq:main-gaussian-energy-certificate}.
This proves that such a state exists, but it does not give an
efficient procedure for finding it.

\paragraph{Using the almost-i.i.d. theorem.}
The preceding proof keeps only the $t$ replicas entering the
objective. We can also first approximate a larger reduced state
on $k$ replicas, and then keep $t$ of them.
Let $\rho^\ell$ satisfy the same constraints.
Write $\ell=k+m$, with $k\geq t$ and $m\geq1$, and set
$\sigma^k:=\Tr_m(\rho^\ell)$.
For $0\leq r\leq k-1$, let $\delta_{k,m,n}(r)$ be any of the
trace-norm bounds in
Theorem~\ref{thm:fermionic-gaussian-definetti} or
Corollary~\ref{cor:fermionic-gaussian-renner-bound}.
Write the operators supplied by the theorem as
$\tilde{\rho}_U^k=w(U)\eta_U^k$, where $\eta_U^k$ has trace one
and is supported on $\Ran(P_U^{k,r})$.
Only terms with $w(U)>0$ need to be considered.

The state $\eta_U^k$ need not be a tensor product.
Its support condition has a simple meaning: measuring each
replica with the projectors $\nu_U$ and $\mathbb{1}-\nu_U$
can give the second outcome on at most $r$ replicas.
If $t$ replicas are selected uniformly, the probability of
obtaining this outcome on at least one of them is at most $rt/k$.
We use this observation to bound the distance between their
reduced state and $\nu_U^{\otimes t}$.

Choosing $t$ replicas uniformly and tracing out the others gives
the map
\begin{equation}
\label{eq:app-averaging-map}
    \mathcal{A}_{k,t}(\eta^k)
    :=
    \binom{k}{t}^{-1}
    \sum_{\substack{S\subseteq\{1,\ldots,k\}\\|S|=t}}
    \Tr_{S^c}(\eta^k).
\end{equation}
The replicas in each subset are ordered increasingly, so all
terms act on the same $t$-replica space.
Set $\eta_{U,t}:=\mathcal{A}_{k,t}(\eta_U^k)$.
The operator
\begin{equation}
    N_U^{(k)}
    :=
    \sum_{i=1}^k(\mathbb{1}-\nu_U)^{(i)}
\end{equation}
counts the number of second outcomes in the measurement above.
Since $P_U^{k,r}$ projects onto eigenvalues $0,\ldots,r$ of
$N_U^{(k)}$, we have $\Tr(N_U^{(k)}\eta_U^k)\leq r$.
For any set $S$ of $t$ replicas, the probability of at least one
second outcome is at most the sum of its probabilities on the
individual replicas. In operator form,
\begin{equation}
    \mathbb{1}-\prod_{i\in S}\nu_U^{(i)}
    \leq
    \sum_{i\in S}(\mathbb{1}-\nu_U)^{(i)}.
\end{equation}
Averaging over $S$ gives
\begin{equation}
\label{eq:app-defect-overlap}
    1-\Tr\left[\nu_U^{\otimes t}\eta_{U,t}\right]
    \leq
    \frac{t}{k}\Tr(N_U^{(k)}\eta_U^k)
    \leq
    \frac{rt}{k}.
\end{equation}
Since $\nu_U^{\otimes t}$ is pure, the Fuchs--van de Graaf
inequality converts this overlap bound into
\begin{equation}
\label{eq:app-defect-removal}
    \frac12\left\|\eta_{U,t}-\nu_U^{\otimes t}\right\|_1
    \leq
    s_{k,t}(r)
    :=
    \sqrt{\min\left\{1,\frac{rt}{k}\right\}}.
\end{equation}

We can now replace the reduced states $\eta_{U,t}$ by the
Gaussian tensor powers $\nu_U^{\otimes t}$.
Let $p:=\int w(U)\,dU\leq1$ and define
\begin{equation}
    \gamma_t
    :=
    \int w(U)\nu_U^{\otimes t}\,dU
    +
    (1-p)\nu_0^{\otimes t}.
\end{equation}
The last term makes the total weight equal to one.
There are two contributions to the error.
First, applying the de Finetti theorem and adding the remaining
weight to $\nu_0^{\otimes k}$ gives a trace distance at most
$\delta_{k,m,n}(r)$, by
Equation~\eqref{eq:app-normalized-error}.
Choosing $t$ replicas and tracing out the others cannot increase
this distance.
Second, replacing each $\eta_{U,t}$ by $\nu_U^{\otimes t}$
adds at most $s_{k,t}(r)$ times its weight.
Since $\sigma^k$ is permutation symmetric,
$\mathcal{A}_{k,t}(\sigma^k)=\Tr_{\ell-t}(\rho^\ell)$.
The triangle inequality gives
\begin{align}
\label{eq:app-almost-iid-local-error}
    \frac12\left\|
        \Tr_{\ell-t}(\rho^\ell)-\gamma_t
    \right\|_1
    &\leq
    \delta_{k,m,n}(r)+p\,s_{k,t}(r)
    \nonumber\\
    &\leq
    \delta_{k,m,n}(r)+s_{k,t}(r).
\end{align}
Using the same argument for expectation values as in the proof
of Proposition~\ref{prop:gaussian-optimization-gap}, we obtain
\begin{align}
\label{eq:app-almost-iid-gap}
    0
    &\leq
    \omega_\ell^{\mathrm{ext}}(C_t)-\omega_t(C_t)
    \nonumber\\
    &\leq
    \operatorname{osc}(C_t)
    \left[
        \delta_{k,m,n}(r)+s_{k,t}(r)
    \right].
\end{align}
The averaging argument also gives a Gaussian state whose
objective value is within this error of the relaxed maximum.
We may choose $r$, $k$ and $m$ to minimize the bound, subject to
$k+m=\ell$, $k\geq t$, $m\geq1$ and $0\leq r\leq k-1$.
For fixed $t$, the simpler choice $k=t$ and $r=0$ already proves
convergence of the hierarchy.

\paragraph{Additional constraints.}
The results above use the full Gaussian family.
If the optimization is restricted to a fixed parity or particle
number, the Gaussian states in the de Finetti mixture must
satisfy the same restriction. This requires a de Finetti bound
for the corresponding family.
For fixed particle number, the reference state should be a
Slater determinant with the required number of particles.
For a nonzero particle number, a number preserving Gaussian
unitary cannot produce the required state from the vacuum,
since it leaves the vacuum unchanged.

A linear constraint on the extension need not hold for each
Gaussian state in the resulting mixture.
For example,
$\int\Tr(A\nu_U)\,d\mu(U)=b$
does not imply $\Tr(A\nu_U)=b$ for every $U$.
To estimate how closely the individual states satisfy the
constraint, consider the squared error. With
$B:=A-b\mathbb{1}$,
\begin{equation}
    \left(\Tr(B\nu_U)\right)^2
    =
    \Tr\left[(B\otimes B)\nu_U^{\otimes2}\right].
\end{equation}
If its average is at most $\varepsilon^2$, the Markov inequality
gives, for every $\alpha>0$,
\begin{equation}
    \mu\left\{U:\left|\Tr(B\nu_U)\right|>\alpha\right\}
    \leq
    \frac{\varepsilon^2}{\alpha^2}.
\end{equation}
Thus the probability of choosing a Gaussian state whose
constraint error exceeds $\alpha$ is at most
$\varepsilon^2/\alpha^2$.
To guarantee a state with both a small constraint error and an
objective value close to the maximum, these two requirements
must be considered together.

\paragraph{Mixed Gaussian states.}
The result extends to mixed Gaussian states by purification.
Every mixed Gaussian state $\sigma$ on $\Hcal$ can be obtained
by tracing out $\Kcal$ from a pure Gaussian state $\phi$ on
$\Hcal\otimes\Kcal$, where $\Kcal\cong\Hcal$.
Let $\bar C_t$ act as $C_t$ on the $\Hcal$ factors and as the
identity on the $\Kcal$ factors. Then
\begin{equation}
    \Tr\left[\bar C_t\phi^{\otimes t}\right]
    =
    \Tr\left[C_t(\Tr_{\Kcal}\phi)^{\otimes t}\right].
\end{equation}
The two objective values are identical.
Conversely, tracing out $\Kcal$ from any pure Gaussian $\phi$
gives a Gaussian state.
We can thus optimize over mixed Gaussian states on $n$ modes
by optimizing over pure Gaussian states on $2n$ modes.
Adding the identity factors does not change the largest or
smallest eigenvalue of the objective.
Proposition~\ref{prop:gaussian-optimization-gap} applies with
$L_n$ replaced by $L_{2n}$.
By Theorem~\ref{thm:gaussian-invariant}, the same optimization
over enlarged replicas can also be written using
Gaussian-invariant extensions on the original replicas.

For an unconstrained energy minimization, allowing mixed
Gaussian states does not improve the optimum.
Every mixed Gaussian state is a mixture of pure Gaussian
states~\cite{Vershynina_2014}, and its energy is the average
of their energies. At least one pure state in the mixture has
no larger energy.
The purification construction is useful for objectives that
are nonlinear in the density matrix.

\paragraph{Computational cost.}
Let $P_\ell$ project onto $\GSym^\ell(\Hcal)$.
Without additional constraints, finding
$\omega_\ell^{\mathrm{ext}}(C_t)$ amounts to finding the largest
eigenvalue of
\begin{equation}
    P_\ell
    \left(C_t\otimes\mathbb{1}^{\otimes(\ell-t)}\right)
    P_\ell
\end{equation}
acting on $\GSym^\ell(\Hcal)$.
For the energy objective $C_1=H$, minimization uses the smallest
eigenvalue instead.
In a basis of this subspace, this is a
$D_\ell\times D_\ell$ matrix, where
\begin{equation}
    D_\ell
    =
    2\prod_{1\leq i<j\leq n}
    \frac{\ell+2n-i-j}{2n-i-j}
    =
    \mathcal{O}(\ell^{L_n})
\end{equation}
for fixed $n$.
The dimension grows polynomially with the number of replicas
$\ell$, but the exponent $L_n=n(n-1)/2$ grows quadratically
with the number of modes.
A numerical implementation needs a basis of
$\GSym^\ell(\Hcal)$ and the matrix elements of the operator
above in that basis.
Additional linear constraints give a semidefinite program
for density matrices on the same subspace.
These observations do not establish an algorithm whose cost
is polynomial in $n$, or an efficient method for finding
the Gaussian state whose existence follows from the bound.

\section{Note on Graded and Ungraded Conventions}

In our discussion we have used the ungraded Majorana convention, where the operators $\gamma_\mu^{(a)}$ anticommute within the same replica, but commute across distinct replicas. This is in contrast to \cite{Sierant_2026, braccia_2026} where an explicitly graded convention is used, meaning that the Majorana operators $\gamma_\mu^{(a)}$ anticommute across distinct replicas. Here we provide a short remark which makes our results applicable to both conventions. In particular, the Gaussian-symmetric subspace $\GSym^k(\Hcal)$ is the same regardless of convention, meaning that Theorem \ref{thm:fermionic-gaussian-definetti} applies to both. Furthermore, the ungraded bridge operators $\Lambda_{ab}$ and their graded versions $\widehat{\Lambda}_{ab}$ are related by a simple parity string. This allows us to apply Theorem \ref{thm:gaussian-invariant} (which we prove in the ungraded convention for simplicity) within the graded convention.

\begin{remark}
\label{rem:graded-ungraded-bridges}
Let $\Hcal \cong (\mathbb{C}^{2})^{\otimes n}$, and let
$\{\gamma_\mu\}_{\mu=1}^{2n}$ be Hermitian Majorana operators on
$\Hcal$ satisfying
\begin{equation}
    \{\gamma_\mu,\gamma_\nu\}  = 2\delta_{\mu\nu}\mathbb{1}.
\end{equation}
Furthermore, define the local parity operator:
\begin{equation}
    \Gamma := (-i)^n\prod_{\mu=1}^{2n}\gamma_\mu,
\end{equation}
satisfying $\Gamma^\dagger=\Gamma$, $\Gamma^2=\mathbb{1}$ and
$\{\Gamma,\gamma_\mu\}=0$ for every $\mu$.

For a tensor product space $\Hcal^{\otimes k}$, let $\gamma_\mu^{(a)}$ denote the \textbf{ungraded} Majorana operators acting as $\gamma_\mu$ on the $a$-th replica, and as identity on all other replicas. For two replicas with $a \neq b$, these operators satisfy:
\begin{align}
[\gamma_\mu^{(a)},\gamma_\nu^{(b)}] &= 0, & [\Gamma^{(a)},\Gamma^{(b)}] &= 0, & [\Gamma^{(a)},\gamma_\mu^{(b)}] &= 0.
\end{align}
Furthermore, the ungraded bridge operators $\Lambda_{ab}$ are defined as:
\begin{equation}
    \Lambda_{ab} := \sum_{\mu=1}^{2n} \gamma_\mu^{(a)}\gamma_\mu^{(b)}, \quad 1\leq a<b\leq k.
\end{equation}
For the same tensor product space $\Hcal^{\otimes k}$, the \textbf{graded} Majorana operators $\widehat{\gamma}_\mu^{(a)}$ are defined as:
\begin{equation} \label{eq:graded-majorana-def}
    \widehat{\gamma}_\mu^{(a)} := \Bigl( \prod_{c<a}\Gamma^{(c)} \Bigr)\gamma_\mu^{(a)}.
\end{equation}
They satisfy the replicated canonical anticommutation relations:
\begin{equation} \label{eq:graded-replicated-anticommutation}
    \left\{ \widehat{\gamma}_\mu^{(a)},  \widehat{\gamma}_\nu^{(b)} \right\} =  2\delta_{ab}\delta_{\mu\nu}\mathbb{1},
\end{equation}
where the Hermitian, graded bridge operators are defined as:
\begin{equation}
    \widehat{\Lambda}_{ab} := i\sum_{\mu=1}^{2n} \widehat{\gamma}_\mu^{(a)} \widehat{\gamma}_\mu^{(b)}, \quad 1\leq a<b\leq k.
\end{equation}
For each $a<b$, introduce the intervening parity string:
\begin{equation}
    R_{ab} := \prod_{c=a}^{b-1}\Gamma^{(c)},
\end{equation}
which satisfies $R_{ab}^{\dagger}=R_{ab}$ and $R_{ab}^{2}=\mathbb{1}$. The graded and ungraded bridge operators are related as:
\begin{align}
\label{eq:graded-ungraded-bridge-conversion}
    \widehat{\Lambda}_{ab}  &=  i\Lambda_{ab}R_{ab}, &
    \Lambda_{ab} &=  -i\widehat{\Lambda}_{ab}R_{ab}.
\end{align}
Moreover,
\begin{align} \label{eq:graded-ungraded-anticommutation}
    \{\Lambda_{ab},R_{ab}\} &=0, &
    \{\widehat{\Lambda}_{ab},R_{ab}\} &= 0.
\end{align}
In particular, the ordinary and graded bridges have the same kernel:
\begin{equation}
\label{eq:graded-ungraded-kernel}
    \ker\widehat{\Lambda}_{ab}  =  \ker\Lambda_{ab}.
\end{equation}
As a consequence, the Gaussian-symmetric subspace is the same in both conventions:
\begin{equation}
\label{eq:graded-ungraded-gsym}
    \GSym^k(\Hcal) = \widehat{\GSym}^k(\Hcal).
\end{equation}
Finally, for any operator $X$ on $\Hcal^{\otimes k}$,
\begin{equation}
\label{eq:graded-ungraded-commutation}
    [X,\Lambda_{ab}]=0 \quad \Longleftrightarrow \quad [X,\widehat{\Lambda}_{ab}R_{ab}]=0.
\end{equation}
It also follows that if $[X, R_{ab}] = 0$, then $[X,\Lambda_{ab}]=0$ if and only if $[X,\widehat{\Lambda}_{ab}]=0$.
\end{remark}

\begin{proof}
The intervening parity string construction is the usual Jordan--Wigner representation across replicas \cite{braccia_2026}. Within a single replica, the local canonical anticommutation relations are unchanged and coincide for both conventions. On the other hand, if $a<b$, then the string $R_{ab}$ contains $\Gamma^{(a)}$, which anticommutes with $\gamma_\mu^{(a)}$. This gives the required minus sign for anticommutation across replicas, leading to Equation \eqref{eq:graded-replicated-anticommutation}. To show Equation \eqref{eq:graded-ungraded-bridge-conversion}, assume $a < b$ and write out each term using Equation \eqref{eq:graded-majorana-def}:
\begin{align}
    \widehat{\gamma}_\mu^{(a)}  \widehat{\gamma}_\mu^{(b)}
    &=
    \bigl( \prod_{c < a} \Gamma^{(c)} \bigr) \gamma_\mu^{(a)} \bigl(\prod_{d < b} \Gamma^{(d)} \bigr) \gamma_\mu^{(b)}
    \\
    &=
    \gamma_\mu^{(a)} \bigl( \prod_{c < a} \Gamma^{(c)} \bigr) \bigl( \prod_{d < b} \Gamma^{(d)} \bigr) \gamma_\mu^{(b)}
    \\
    &=
    \gamma_\mu^{(a)} \bigl(\prod_{c = a}^{b-1} \Gamma^{(c)} \bigr) \gamma_\mu^{(b)} \\
    &= \gamma_\mu^{(a)}  \gamma_\mu^{(b)} \bigl( \prod_{c = a}^{b-1} \Gamma^{(c)} \bigr) = \gamma_\mu^{(a)}  \gamma_\mu^{(b)} R_{ab},
\end{align}
where we have used the fact that in the second line, all factors of $\Gamma^{(c <a)}$ square to identity. Finally, no term in the product on the fourth line acts on the $b$-th replica, meaning we can commute it to the right hand side. Summing over $\mu$ then gives the left hand side of Equation \eqref{eq:graded-ungraded-bridge-conversion}, whereas the right hand side follows from $R_{ab}^2 = \mathbb{1}$. Equation \eqref{eq:graded-ungraded-anticommutation} then follows from taking Hermitian conjugates.

\noindent Next, suppose that $\Lambda_{ab}|\psi\rangle=0$. Then:
\begin{equation}
    \widehat{\Lambda}_{ab}|\psi\rangle  =  i\Lambda_{ab}R_{ab}|\psi\rangle   =  -iR_{ab}\Lambda_{ab}|\psi\rangle    =  0.
\end{equation}
Conversely, if $\widehat{\Lambda}_{ab}|\psi\rangle=0$, then:
\begin{equation}
\Lambda_{ab}|\psi\rangle = -i\widehat{\Lambda}_{ab}R_{ab}|\psi\rangle = iR_{ab}\widehat{\Lambda}_{ab}|\psi\rangle = 0,
\end{equation} 
proving Equation \eqref{eq:graded-ungraded-kernel}. Taking the intersection of the kernels over all pairs $(a,b)$ then proves Equation \eqref{eq:graded-ungraded-gsym}. Finally, Equation \eqref{eq:graded-ungraded-commutation} follows from Equation \eqref{eq:graded-ungraded-bridge-conversion} and bilinearity of the commutator. Writing out the right hand side as $X \widehat{\Lambda}_{ab} = \widehat{\Lambda}_{ab} R_{ab} X R_{ab}$ establishes the final claim.
\end{proof}

\begin{remark} \label{rem:graded-ungraded-theorems} Theorem \ref{thm:fermionic-gaussian-definetti} and Corollary \ref{cor:fermionic-gaussian-renner-bound} apply to the graded convention verbatim, where $\Lambda_{ab}$ is replaced by $\widehat{\Lambda}_{ab}$ and the Gaussian-symmetric subspace $\GSym^k(\Hcal)$ is replaced by its graded counterpart $\widehat{\GSym}^k(\Hcal)$. In Theorem \ref{thm:gaussian-invariant}, one may obtain the statement for the graded convention by setting $\Lambda_{ab} = -i \widehat{\Lambda}_{ab} R_{ab}$. 
\end{remark}

We note that compared to the ungraded convention, partial traces $\rho = \Tr_{\Kcal^{\otimes k}}(\bar{\rho})$ of Gaussian-symmetric states $\bar{\rho} \in \widehat{\GSym}^k(\Hcal \otimes \Kcal)$ do not in general commute with $\widehat{\Lambda}_{ab}$. To see this, take $n=1, k=2$. For the triplet state $|T \rangle = (|01\rangle + |10 \rangle)/\sqrt{2}$, define the state:
\begin{equation}
|\bar{\psi} \rangle = |T \rangle_{\Hcal_1 \Hcal_2} \otimes |T \rangle_{\Kcal_1 \Kcal_2}.
\end{equation}
The graded, enlarged bridge operator $\bar{\widehat{\Lambda}}_{12}$ takes the form:
\begin{equation}
-i \bar{\widehat{\Lambda}}_{12} 
= (X_{\Hcal_1})(Z_{\Hcal_1} Z_{\Kcal_1} X_{\Hcal_2})  
+ (Y_{\Hcal_1})(Z_{\Hcal_1} Z_{\Kcal_1} Y_{\Hcal_2}) 
+ (Z_{{\Hcal_1}} X_{\Kcal_1})(Z_{\Hcal_1} Z_{\Kcal_1} Z_{\Hcal_2} X_{\Kcal_2}) 
+ (Z_{\Hcal_1} Y_{\Kcal_1})(Z_{\Hcal_1} Z_{\Kcal_1} Z_{\Hcal_2} Y_{\Kcal_2}), 
\end{equation}
which simplifies into:
\begin{align}
-i \bar{\widehat{\Lambda}}_{12} &= \Big[ (X_{\Hcal_1} X_{\Hcal_2} + Y_{\Hcal_1} Y_{\Hcal_2}) + (Z_{\Hcal_1} Z_{\Hcal_2}) (X_{\Kcal_1} X_{\Kcal_2} + Y_{\Kcal_1} Y_{\Kcal_2}) \Big] Z_{\Hcal_1} Z_{\Kcal_1} \\
&= \Big[ \Lambda_{12}^{\Hcal_1 \Hcal_2} \otimes \mathbb{1} + Q_{12}^{\Hcal_1 \Hcal_2} \otimes \Lambda_{12}^{\Kcal_1 \Kcal_2} \Big] Z_{\Hcal_1} Z_{\Kcal_1}.
\end{align}
Now write $|S\rangle = (|01 \rangle - |10 \rangle)/\sqrt{2}$. We have that:
\begin{align}
(Z \otimes \mathbb{1})|T \rangle &= |S \rangle \\
(X \otimes X + Y \otimes Y)|S \rangle &= -2|S \rangle \\
(X \otimes X + Y \otimes Y)|T \rangle &= 2|T \rangle \\
(Z \otimes Z) |S \rangle &= -|S \rangle.
\end{align}
It follows that:
\begin{equation}
\bar{\widehat{\Lambda}}_{12} |\bar{\psi} \rangle = i(-2 + 2) |S \rangle_{\Hcal_1 \Hcal_2} \otimes |S \rangle_{\Kcal_1 \Kcal_2} =  0,
\end{equation}
hence $|\bar{\psi} \rangle$ is Gaussian-symmetric. Its partial trace is given by:
\begin{equation}
\Tr_{\Kcal_1 \Kcal_2}(| \bar{\psi} \rangle \langle \bar{\psi}  |)  = |T \rangle \langle T| = | \psi \rangle \langle \psi |,
\end{equation}
and on $\Hcal \otimes \Hcal$ the graded operator $\widehat{\Lambda}_{12}$ becomes:
\begin{align}
-i \widehat{\Lambda}_{12} &= (X_{\Hcal_1}) (Z_{\Hcal_1} X_{\Hcal_2}) + (Y_{\Hcal_1}) (Z_{\Hcal_1} Y_{\Hcal_2}) \\
&= (X_{\Hcal_1} X_{\Hcal_2} + Y_{\Hcal_1} Y_{\Hcal_2}) Z_{\Hcal_1}.
\end{align}
We then have that: 
\begin{equation}
\widehat{\Lambda}_{12} |\psi \rangle = i (X \otimes X + Y \otimes Y)|S \rangle = -2i |S \rangle,
\end{equation}hence $\big[ |\psi\rangle \langle \psi |,\widehat{\Lambda}_{12} \big] \neq 0$. On the other hand, $|\psi\rangle \langle \psi |$ commutes with the ungraded bridge:
\begin{equation}
\Lambda_{12} = X_{\Hcal_1} X_{\Hcal_2} + Y_{\Hcal_1} Y_{\Hcal_2},
\end{equation}
since we have $\Lambda_{12} |\psi \rangle = 2|\psi\rangle$. Finally, $R_{12} = Z_{\Hcal_1}$, giving $\widehat{\Lambda}_{12} = i \Lambda_{12} R_{12}$ in agreement with Remark \ref{rem:graded-ungraded-bridges}.

\section{Proof of Theorem \ref{thm:fermionic-gaussian-definetti}} 
In this section, we prove Theorem \ref{thm:fermionic-gaussian-definetti} and Corollary \ref{cor:fermionic-gaussian-renner-bound}, which provide an exponential bound on the approximation of the partial trace of a Gaussian-symmetric state by a convex combination of almost-i.i.d. states. We follow the general structure of the proof of the exponential de Finetti theorem for symmetric states, as presented in \cite{Renner_2007}, which gives a bound $\delta \leq 3 D_m \gamma$. The resulting expressions for $\delta$ in Theorem \ref{thm:fermionic-gaussian-definetti} and Corollary \ref{cor:fermionic-gaussian-renner-bound} then follow from bounds on $\gamma$ specific to the FLO setting. As the latter requires more complicated analysis, we prove it separately. Before proceeding to the proof of Theorem \ref{thm:fermionic-gaussian-definetti}, we first present our required technical lemmas.

\begin{lemma}[Lemma 3 of \cite{Renner_2007}] \label{lemma:projectors}
    Let $\{ \rho_\tau \}_{\tau \in \mathcal{T} }$ be a family of nonnegative operators on a Hilbert space $\mathcal{H}$ and let $\{ P_\tau \}_{\tau \in \mathcal{T}}$ be a family of projectors on $\mathcal{H}$. Then, for any measure $d \tau$ on $\mathcal{T}$:
    
    \[
    \left \| \int (\rho_\tau - P_\tau \rho_\tau P_\tau ) d \tau \right \|_1 \leq 3 \left \| \int (\mathbb{1} - P_\tau)\rho_\tau d\tau \right \|_1
    \]
\end{lemma}

\begin{proof}
Set $Q_\tau:=\mathbb{1} -P_\tau$. Then,
\begin{equation}
\rho_\tau-P_\tau\rho_\tau P_\tau
=
Q_\tau\rho_\tau+\rho_\tau Q_\tau-Q_\tau\rho_\tau Q_\tau.
\end{equation}
Hence, with:
\begin{equation}
X:=\int_{\mathcal T}Q_\tau\rho_\tau\,d\tau,
\qquad
Y:=\int_{\mathcal T}Q_\tau\rho_\tau Q_\tau\,d\tau,
\end{equation}
we obtain:
\begin{equation}
\int_{\mathcal T}
\left(
\rho_\tau-P_\tau\rho_\tau P_\tau
\right)d\tau
=
X+X^\dagger-Y.
\end{equation}
Therefore, it holds that:
\begin{equation}
\left\|
\int_{\mathcal T}
\left(
\rho_\tau-P_\tau\rho_\tau P_\tau
\right)d\tau
\right\|_1
\leq
2\|X\|_1+\|Y\|_1.
\end{equation}
Since $Y\geq0$, the trace norm of $Y$ satisfies $\|Y\|_1=\Tr Y$. Moreover,
\begin{align}
\Tr Y &= \int_{\mathcal T}\Tr(Q_\tau\rho_\tau Q_\tau)\,d\tau
\\
&= \int_{\mathcal T}\Tr(Q_\tau\rho_\tau)\,d\tau \\
& = \Tr X \leq \|X\|_1.
\end{align}
Thus, we obtain the desired bound:
\begin{equation}
\left\|
\int_{\mathcal T}
\left(
\rho_\tau-P_\tau\rho_\tau P_\tau
\right)d\tau
\right\|_1
\leq
3\|X\|_1.
\end{equation}
\end{proof}

\begin{lemma}[Gaussian-symmetric Multiplicity]
\label{lem:multiplicity-one-gsym}
Let \(\Mcal_n\) denote the full matchgate group, following the
\(O(2n)\) convention of \cite{Sierant_2026}. Let
\(P_{\GSym^k(\Hcal)}\) be the projector onto the Gaussian-symmetric subspace:
\[
    \GSym^k(\Hcal)
    =
    \operatorname{span}
    \left\{
        (U|\mathbf{0}\rangle)^{\otimes k}:U\in\Mcal_n
    \right\}.
\]
Then, the Gaussian-symmetric irrep occurs with multiplicity one in the diagonal matchgate representation. Equivalently,
\begin{equation}
\label{eq:multiplicity-one-block}
    P_{\GSym^k(\Hcal)}
    \Com_k(\mathfrak{M}_n)
    P_{\GSym^k(\Hcal)}
    =
    \mathbb C P_{\GSym^k(\Hcal)}.
\end{equation}
Consequently, if $X\in \Com_k(\mathfrak{M}_n)$, then:
\begin{equation}
\label{eq:commutant-scalar-on-P0}
    XP_{\GSym^k(\Hcal)}
    =
    P_{\GSym^k(\Hcal)}XP_{\GSym^k(\Hcal)}
    =
    \lambda_X P_{\GSym^k(\Hcal)}
\end{equation}
for some scalar \(\lambda_X\).
\end{lemma}

\begin{proof}
The replicated matchgate representation has a commuting replica
\(SO(k)\) symmetry generated by the bridge operators $\widehat{\Lambda}_{ab}$, where $1\leq a<b\leq k$. The corresponding Gelfand--Tsetlin decomposition labels the blocks of the
matchgate commutant by irreducible \(SO(k)\) highest weights $\nu$.
Schematically, the dual-pair decomposition has the form
\begin{equation}
\label{eq:dual-pair-decomp}
    \Hcal^{\otimes k}
    \cong
    \bigoplus_{\nu}
    W_\nu^{\Mcal_n}
    \otimes
    V_\nu^{SO(k)},
\end{equation}
where \(W_\nu^{\Mcal_n}\) carries the diagonal matchgate action and
\(V_\nu^{SO(k)}\) is the corresponding replica \(SO(k)\) multiplicity
space. Hence:
\begin{equation}
\label{eq:commutant-dual-pair}
    \Com_k(\mathfrak{M}_n)
    \cong
    \bigoplus_\nu
    \mathbb{1}_{W_\nu^{\Mcal_n}}
    \otimes
    \operatorname{End}
    \left(
        V_\nu^{SO(k)}
    \right).
\end{equation}
The replicated vacuum \(|\mathbf{0}\rangle^{\otimes k}\) is annihilated by all
bridge operators \(\widehat{\Lambda}_{ab}\). Therefore it belongs to the trivial
\(SO(k)\) sector, whose highest weight is $\nu = 0$.
The projector \(P_{\GSym^k(\Hcal)}\) is precisely the projector onto this trivial \(SO(k)\) sector.  For the trivial \(SO(k)\) sector we have:
\[
    V_0^{SO(k)}\cong \mathbb C.
\]
Hence,
\[
    \operatorname{End}\left(V_0^{SO(k)}\right)
    =
    \operatorname{End}(\mathbb C)
    =
    \mathbb C.
\]
Thus, the commutant block supported on \(P_{\GSym^k(\Hcal)}\) is one-dimensional:
\[
    P_{\GSym^k(\Hcal)}
    \Com_k(\mathfrak{M}_n)
    P_{\GSym^k(\Hcal)}
    =
    \mathbb C P_{\GSym^k(\Hcal)}.
\]
This proves the multiplicity-one statement. Finally, since \(P_{\GSym^k(\Hcal)}\) is the full isotypic projector for the
\(\nu=0\) block, every element of the commutant preserves
\(\operatorname{Ran} (P_{\GSym^k(\Hcal)}) \). Therefore, $[X,P_{\GSym^k(\Hcal)}]=0$ 
for every commutant element \(X\). Combining this with the
one-dimensionality of the block gives:
\[
    XP_{\GSym^k(\Hcal)}
    =
    P_{\GSym^k(\Hcal)}XP_{\GSym^k(\Hcal)}
    =
    \lambda_XP_{\GSym^k(\Hcal)}.
\]
\end{proof}

\begin{lemma}[Matchgate Twirls]
\label{lem:scalar-action-gsym}
Let \(A\in\mathcal B(\Hcal^{\otimes k})\), and define the twirl:
\begin{equation}
    T_A^{(k)}
    :=
    \int_{\Mcal_n}
    U^{\otimes k}A(U^\dagger)^{\otimes k}\,dU.
\end{equation}
Then,
\begin{align}
    T_A^{(k)}P_{\GSym^k(\Hcal)}
    &=
    P_{\GSym^k(\Hcal)}T_A^{(k)}P_{\GSym^k(\Hcal)} \label{eq:scalar-action-TA} \\
    & =
    \frac{\Tr(P_{\GSym^k(\Hcal)}A)}{D_k}P_{\GSym^k(\Hcal)},
\end{align}
where $D_k:=\Tr ( P_{\GSym^k(\Hcal)} )$. If \(\Tr(P_{\GSym^k(\Hcal)}A)\neq0\) holds, we may define:
\begin{equation}
    \Gamma_A^{(k)}
    :=
    \frac{D_k}{\Tr(P_{\GSym^k(\Hcal)}A)}
    T_A^{(k)},
\end{equation}
and we have that:
\begin{equation}
    \Gamma_A^{(k)}P_{\GSym^k(\Hcal)}
    =
    P_{\GSym^k(\Hcal)}.
\end{equation}
\end{lemma}

\begin{proof}
By Haar invariance, \(T_A^{(k)}\) commutes with the diagonal matchgate action, so $T_A^{(k)} \in \Com_k(\mathfrak{M}_n)$. By Lemma~\ref{lem:multiplicity-one-gsym}, \(T_A^{(k)}\) acts as a scalar on the Gaussian-symmetric block:
\begin{equation}
    T_A^{(k)}P_{\GSym^k(\Hcal)}
    =
    \lambda_A P_{\GSym^k(\Hcal)}. \label{eq:TA-projector}
\end{equation}
Taking the trace gives:
\[
    \lambda_A D_k
    =
    \Tr(P_{\GSym^k(\Hcal)}T_A^{(k)}).
\]
Now use the definition of \(T_A^{(k)}\), cyclicity of the trace, and $[P_{\GSym^k(\Hcal)},U^{\otimes k}]=0$ to obtain: 
\begin{align}
    \Tr(P_{\GSym^k(\Hcal)}T_A^{(k)})
    &=
    \int_{\Mcal_n}
    \Tr\left[
        P_{\GSym^k(\Hcal)}
        U^{\otimes k}A(U^\dagger)^{\otimes k}
    \right]dU
    \\
    &=
    \int_{\Mcal_n}
    \Tr\left[
        (U^\dagger)^{\otimes k}
        P_{\GSym^k(\Hcal)}
        U^{\otimes k}
        A
    \right]dU
    \\
    &=
    \Tr(P_{\GSym^k(\Hcal)}A).
\end{align}
Hence,
\[
    \lambda_A
    =
    \frac{\Tr(P_{\GSym^k(\Hcal)}A)}{D_k}.
\]
This proves Equation \eqref{eq:scalar-action-TA}. The normalized statement follows immediately.
\end{proof}

\begin{lemma}[Partial Projections]
\label{lem:support-partial-projection}
Let \(k,m\geq 1\). Then,
\begin{equation}
\label{eq:K-inclusion-prop}
    \GSym^{k+m}(\Hcal)
    \subseteq
    \Hcal^{\otimes k} \otimes \GSym^m(\Hcal).
\end{equation}
Consequently, if a density operator
\(\rho^{k+m}\in\mathcal B(\Hcal^{\otimes(k+m)})\) is supported on
\(\GSym^{k+m}(\Hcal)\), namely:
\begin{equation}
\label{eq:rho-supported-Kkm}
    \rho^{k+m}
    =
    P_{\GSym^{k+m}(\Hcal)}\rho^{k+m}P_{\GSym^{k+m}(\Hcal)},
\end{equation}
then:
\begin{equation}
\label{eq:left-partial-projection}
    \left(
        \mathbb{1}^{\otimes k}\otimes P_{\GSym^m(\Hcal)}
    \right)
    \rho^{k+m}
    =
    \rho^{k+m},
\end{equation}
and
\begin{equation}
\label{eq:right-partial-projection}
    \rho^{k+m}
    \left(
        \mathbb{1}^{\otimes k}\otimes P_{\GSym^m(\Hcal)}
    \right)
    =
    \rho^{k+m}.
\end{equation}
Equivalently,
\begin{equation}
\label{eq:two-sided-partial-projection}
    \rho^{k+m}
    =
    \left(
        \mathbb{1}^{\otimes k}\otimes P_{\GSym^m(\Hcal)}
    \right)
    \rho^{k+m}
    \left(
        \mathbb{1}^{\otimes k}\otimes P_{\GSym^m(\Hcal)}
    \right).
\end{equation}
\end{lemma}

\begin{proof}
By definition, the Gaussian-symmetric subspace on \(k+m\) replicas is given by:
\begin{equation}
\label{eq:Kkm-definition-proof}
    \GSym^{k+m}(\Hcal)
    =
    \operatorname{span}
    \left\{
        (U|\mathbf{0}\rangle)^{\otimes(k+m)}
        :
        U\in\Mcal_n
    \right\}.
\end{equation}
Every generator of this span factorizes with respect to the splitting
\(\Hcal^{\otimes(k+m)}=\Hcal^{\otimes k}\otimes\Hcal^{\otimes m}\) as:
\begin{equation}
\label{eq:generator-factorization-proof}
    (U|\mathbf{0}\rangle)^{\otimes(k+m)}
    =
    (U|\mathbf{0}\rangle)^{\otimes k}
    \otimes
    (U|\mathbf{0}\rangle)^{\otimes m}.
\end{equation}
For every \(U\in\Mcal_n\), the second tensor factor belongs to
\(\GSym^m(\Hcal)\), because
\begin{equation}
\label{eq:last-factor-in-Km}
    (U|\mathbf{0}\rangle)^{\otimes m}
    \in
    \operatorname{span}
    \left\{
        (V|\mathbf{0}\rangle)^{\otimes m}
        :
        V\in\Mcal_n
    \right\}
    =
    \GSym^m(\Hcal).
\end{equation}
Therefore,
\begin{equation}
\label{eq:generator-in-product-space}
    (U|\mathbf{0}\rangle)^{\otimes(k+m)}
    \in
    \Hcal^{\otimes k}\otimes\GSym^m(\Hcal) \quad \forall U\in\Mcal_n.
\end{equation}
Taking the linear span over all \(U\) gives:
\begin{equation}
    \GSym^{k+m}(\Hcal)
    \subseteq
    \Hcal^{\otimes k}\otimes\GSym^m(\Hcal),
\end{equation}
which proves \eqref{eq:K-inclusion-prop}. Now suppose \(\rho^{k+m}\) is supported on \(\GSym^{k+m}(\Hcal)\). By the
inclusion just proved, the support of \(\rho^{k+m}\) is contained in
\(\Hcal^{\otimes k}\otimes\GSym^m(\Hcal)\). The orthogonal projector on
\(\Hcal^{\otimes k}\otimes\GSym^m(\Hcal)\) is given by:
\begin{equation}
\label{eq:partial-projector-product}
    \mathbb{1}^{\otimes k}\otimes P_{\GSym^m(\Hcal)}.
\end{equation}
Hence, this projector acts as the identity on the support of
\(\rho^{k+m}\). Therefore:
\begin{equation}
    \left(
        \mathbb{1}^{\otimes k}\otimes P_{\GSym^m(\Hcal)}
    \right)
    \rho^{k+m}
    =
    \rho^{k+m}, \label{eq:left-partial-projection-proof}
\end{equation}
and similarly
\begin{equation}
    \rho^{k+m}
    \left(
        \mathbb{1}^{\otimes k}\otimes P_{\GSym^m(\Hcal)}
    \right)
    =
    \rho^{k+m}.
\end{equation}
Combining the two identities gives Equation
\eqref{eq:two-sided-partial-projection}.
\end{proof}

We are now ready to prove Theorem \ref{thm:fermionic-gaussian-definetti}.

\begin{proof}[Proof of Theorem \ref{thm:fermionic-gaussian-definetti}]
For a density operator $\rho^{k +m}$ supported on $\GSym^{k+m}(\Hcal)$, we define the following operators:
\begin{align}
\rho^k_U &:= D_m \times \mathrm{Tr}_m \label{def:rhok}(\mathbb{1}^{\otimes k} \otimes P^{m, 0}_U \cdot \rho^{k + m})  \geq 0 \\
\tilde{\rho}_U^k &:= P_U^{k, r} \rho^k_U P_U^{k, r} \geq 0, \label{def:rhokbar}
\end{align}
where $D_m:=\Tr ( P_{\GSym^m(\Hcal)} )$, and $P_U^{k, r}$ are as given in Definition \ref{def:iid}. The positivity of $\rho^k_U$ (and also of $\tilde{\rho}_U^k$) follows from the fact that $\rho^{k + m}$ is positive and $P^{m, 0}_U$ is a rank one projector. If $\rho^{k + m}$ is rank one, then so is $\rho^k_U$, and hence $\tilde{\rho}_U^k$ is of rank at most one.

We begin by integrating $\rho^k_U$ over the matchgate group:
\begin{align}
\int \rho^k_U dU &= D_m \times \int \mathrm{Tr}_m (\mathbb{1}^{\otimes k} \otimes P^{m, 0}_U \cdot \rho^{k + m})  dU \\
&= \mathrm{Tr}_m \left[ \left( \mathbb{1}^{\otimes k} \otimes D_m \int P^{m, 0}_U dU \right) \rho^{k + m} \right]. \label{eq:rho_integral}\\
\end{align}
By Definition \ref{def:iid} we have $P^{m, 0}_U = U^{\otimes m} P_{\mathbb{1}}^{m, 0} (U^\dag)^{\otimes m}$, where $P_{\mathbb{1}}^{m, 0} = \nu_0^{\otimes m} = |\mathbf{0} \rangle \langle \mathbf{0} |^{\otimes m}$ is the rank-one projector onto the $m$-replicated vacuum. The matchgate twirl of this projector is an element of $\Com_m(\mathfrak{M}_n)$, and was previously evaluated in Equation (160) in \cite{Sierant_2026}:
\begin{equation}
\label{eq:coherent-resolution}
    T_{\nu_0^{\otimes m}}^{(m)} = \int_{\Mcal_n} (U\nu_0 U^\dagger)^{\otimes m}\,dU = \frac{1}{D_m}P_{\GSym^m(\Hcal)}.
\end{equation}
This is the fermionic Gaussian analogue of the usual Haar coherent-state
resolution of the symmetric projector. Putting everything together, we may write:
\begin{align}
D_m \int P^{m, 0}_U dU &=  D_m \int U^{\otimes m} P_{\mathbb{1}}^{m, 0} (U^\dag)^{\otimes m} dU \\
&= D_m \int U^{\otimes m} |\mathbf{0} \rangle \langle \mathbf{0} |^{\otimes m} (U^\dag)^{\otimes m} dU \\
&= P_{\GSym^m(\Hcal)}.
\end{align}
Hence, by Equation \eqref{eq:left-partial-projection-proof} of Lemma \ref{lem:support-partial-projection} the integral simplifies into:
\begin{align}
\int \rho^k_U dU &= \mathrm{Tr}_m \left[ \left( \mathbb{1}^{\otimes k} \otimes P_{\GSym^m(\Hcal)} \right) \rho^{k + m} \right]\\
&= \mathrm{Tr}_m \left[  \rho^{k + m} \right].
\end{align}

Next, we wish to bound the trace norm between $\mathrm{Tr}_m ( \rho^{k + m})$ and $\int \tilde{\rho}^k_U dU$ (the latter previously defined in Equation \eqref{def:rhokbar}):
\begin{align}
\delta &:= \| \mathrm{Tr}_m(  \rho^{k + m}) - \int \tilde{\rho}^k_U dU \|_1 \\
&= \|\int(\rho_U^k - P_U^{k, r} \rho^k_U P_U^{k, r} ) dU \|_1,
\end{align}
where we applied Lemma \ref{lemma:projectors}. We can write this as:
\begin{align}
\delta &=   \|\int(\rho_U^k - P_U^{k, r} \rho^k_U P_U^{k, r} ) dU \|_1 \\
& \leq 3 \| \int (\mathbb{1}^{\otimes k} - P_U^{k, r}) \rho_U^k dU \|_1\\
& \leq 3 \| \int (P_U^{k, r})^{\bot} \rho_U^k dU \|_1, \label{eq:delta-orthogonal-projector}
\end{align}
where $(P_U^{k, r})^{\bot} = (\mathbb{1}^{\otimes k} - P_U^{k, r})$ is a projector orthogonal to $P_U^{k, r}$. Since $P_U^{k, r} =U^{\otimes k} P_{\mathbb{1}}^{k, r} (U^\dag)^{\otimes k}$ (Definition \ref{def:iid}), we can write:
\[ (P_U^{k, r})^{\bot} = \mathbb{1}^{\otimes k} - U^{\otimes k} P_{\mathbb{1}}^{k, r} (U^\dag)^{\otimes k},\] 

\noindent which lets us write $(P_U^{k, r})^{\bot} =  U^{\otimes k} (P_{\mathbb{1}}^{k, r})^{\bot} (U^\dag)^{\otimes k}$. We use this, along with Equation \eqref{def:rhok} to rewrite the integrand in Equation \eqref{eq:delta-orthogonal-projector}:
\begin{align}
(P_U^{k, r})^{\bot} \rho_U^k 
&= D_m \, (P_U^{k, r})^{\bot}\mathrm{Tr}_m (\mathbb{1}^{\otimes k} \otimes P^{m, 0}_U \, \rho^{k + m}) \\
& = D_m \, \mathrm{Tr}_m ((P_U^{k, r})^{\bot} \otimes P^{m, 0}_U \, \rho^{k + m}) \\
& = D_m \, \mathrm{Tr}_m (U^{\otimes k+m} [(P_{\mathbb{1}}^{k, r})^{\bot} \otimes P^{m, 0}_{\mathbb{1}} ] (U^\dag)^{\otimes k+m} \, \rho^{k+m}).
\end{align}
Hence, we have that:
\begin{align}
    \int (P_U^{k, r})^{\bot}  \rho_U^k dU 
    & = D_m  \mathrm{Tr}_m \left[ \int U^{\otimes k+m} [(P_{\mathbb{1}}^{k, r})^{\bot} \otimes P^{m, 0}_{\mathbb{1}} ] (U^\dag)^{\otimes k+m} dU \, \rho^{k+m} \right] \nonumber\\
&= D_m \mathrm{Tr}_m \left[ T^{(k+m)}_{(P_{\mathbb{1}}^{k, r})^{\bot} \otimes P^{m, 0}_{\mathbb{1}}} \cdot \rho^{k+m} \right] 
\end{align}
where the operator $T^{(k+m)}_{(P_{\mathbb{1}}^{k, r})^{\bot} \otimes P^{m, 0}_{\mathbb{1}}}$ is the matchgate twirl of $(P_{\mathbb{1}}^{k, r})^{\bot} \otimes P^{m, 0}_{\mathbb{1}}$. Next, we need to upper bound the trace norm:
\begin{equation} 
    \left \|\mathrm{Tr}_m \left( T^{(k+m)}_{(P_{\mathbb{1}}^{k, r})^{\bot} \otimes P^{m, 0}_{\mathbb{1}}} \cdot \rho^{k+m} \right) \right \|_1. 
\end{equation} 
Since $\rho^{k+m}$ is supported on the Gaussian-symmetric subspace, we are free to insert a projector $P_{\GSym^{k+m}(\Hcal)}$ onto the Gaussian-symmetric subspace into the brackets. By Lemma \ref{lem:multiplicity-one-gsym} (or Equation \eqref{eq:TA-projector}) we then have:
\begin{align}
\left \|\mathrm{Tr}_m \left( T^{(k+m)}_{(P_{\mathbb{1}}^{k, r})^{\bot} \otimes P^{m, 0}_{\mathbb{1}}} \cdot \rho^{k+m} \right) \right \|_1 &= \left \|\mathrm{Tr}_m \left( T^{(k+m)}_{(P_{\mathbb{1}}^{k, r})^{\bot} \otimes P^{m, 0}_{\mathbb{1}}} \cdot P_{\GSym^{k+m}(\Hcal)} \cdot \rho^{k+m} \right) \right \|_1 \\
&= \gamma \times \left \|\mathrm{Tr}_m \left( P_{\GSym^{k+m}(\Hcal)} \cdot \rho^{k+m} \right) \right \|_1 \\
&= \gamma \times \Tr(\rho^{k+m}) = \gamma.
\end{align}
Overall, the error $\delta$ satisfies:
\begin{align}
\label{eq:delta_raw}
\delta \leq 3 \times D_m \times \gamma
\end{align}
and $\gamma$ is a function which we will bound, given via Lemma \ref{lem:scalar-action-gsym}:
\begin{equation}
\label{eq:gamma-trace}
\gamma
:=
\frac{1}{D_{k+m}} \Tr\left[
P_{\GSym^{k+m}(\Hcal)}
\left(
(P_{\mathbb{1}}^{k,r})^\perp\otimes P_{\mathbb{1}}^{m,0}
\right)
\right].
\end{equation}

Next, we derive a closed form expression for $\gamma$, which in turn lets us provide an upper bound on $\delta$. To do so, we will use the Gaussian orbit structure to evaluate the trace in Equation \eqref{eq:gamma-trace}. First, we define an additional object, the vacuum-overlap parameter $q_U$ for $U\in\Mcal_n$:
\begin{equation}
\label{eq:qU-def}
    q_U
    :=
    \left|
        \langle \mathbf{0}|U|\mathbf{0}\rangle
    \right|^2,
\end{equation}
where for every integer \(\ell\geq 1\):
\begin{equation}
\label{eq:qU-power-def}
    q_U^\ell
    =
    \left|
        \langle \mathbf{0}|U|\mathbf{0}\rangle
    \right|^{2\ell}.
\end{equation}
The dimension of the Gaussian-symmetric subspace \(\GSym^\ell(\Hcal)\) is given by: 
\begin{equation}
D_\ell := \Tr (P_{\GSym^\ell(\Hcal)}) = \dim (\GSym^\ell(\Hcal)).
\end{equation}
With the full matchgate convention used here, an analytic formula for $D_\ell$ was given in Equation (F7) of \cite{Sierant_2026}:
\begin{equation}
\label{eq:Dell-product}
    D_\ell
    =
    2
    \prod_{1\leq i<j\leq n}
    \frac{\ell+2n-i-j}{2n-i-j}.
\end{equation}
The moments \(q^\ell_U\), where $\ell\geq 1$, are determined by the Gaussian-symmetric dimension:
\begin{align}
    \int_{\Mcal_n} q_U^\ell\,dU
    &=
    \int_{\Mcal_n}
    \Tr\left[
        (| \mathbf{0}\rangle\langle\mathbf{0}|)^{\otimes \ell}
        (U|\mathbf{0}\rangle\langle\mathbf{0}|U^\dagger)^{\otimes\ell}
    \right]dU
    \\
    &=
    \Tr\left[
        (| \mathbf{0}\rangle\langle\mathbf{0}|)^{\otimes \ell}
        \frac{P_{\GSym^\ell(\Hcal)}}{D_\ell}
    \right] \\
    &=
    \frac{1}{D_\ell}, \label{eq:qU-moment}
\end{align}
which holds because $|\mathbf{0}\rangle^{\otimes \ell}\in \GSym^\ell(\mathcal{H})$, or equivalently $(| \mathbf{0}\rangle\langle\mathbf{0}|)^{\otimes \ell} \leq P_{\GSym^\ell(\Hcal)}$.

We are now ready to obtain a bound on $\gamma$. We continue from Equation \eqref{eq:gamma-trace}, rewriting it as a single trace:
\begin{equation}
\label{eq:gamma-normalized-projector}
\gamma
=
\Tr\left[
\frac{P_{\GSym^{k+m}(\Hcal)}}{D_{k+m}}
\left(
(P_{\mathbb{1}}^{k,r})^\perp\otimes P_{\mathbb{1}}^{m,0}
\right)
\right].
\end{equation}
Using the coherent-state resolution in Equation \eqref{eq:coherent-resolution},
\begin{equation}
\label{eq:coherent-resolution-km}
\frac{P_{\GSym^{k+m}(\Hcal)}}{D_{k+m}}
=
\int_{\Mcal_n}
(U|\mathbf{0}\rangle\langle\mathbf{0}|U^\dagger)^{\otimes(k+m)}
\,dU.
\end{equation}
Overall, Equation \eqref{eq:gamma-trace} becomes:
\begin{align}
\gamma
&=
\Tr\left[
\left(
\int_{\Mcal_n}
(U|\mathbf{0}\rangle\langle\mathbf{0}|U^\dagger)^{\otimes(k+m)}
\,dU
\right)
\left(
(P_{\mathbb{1}}^{k,r})^\perp\otimes P_{\mathbb{1}}^{m,0}
\right)
\right]
\\
&=
\int_{\Mcal_n}
\Tr\left[
(U|\mathbf{0}\rangle\langle\mathbf{0}|U^\dagger)^{\otimes(k+m)}
\left(
(P_{\mathbb{1}}^{k,r})^\perp\otimes P_{\mathbb{1}}^{m,0}
\right)
\right]dU .
\label{eq:gamma-integral-first}
\end{align}
Next, we use the factorisation:
\begin{equation}
(U|\mathbf{0}\rangle\langle\mathbf{0}|U^\dagger)^{\otimes(k+m)}
=
(U|\mathbf{0}\rangle\langle\mathbf{0}|U^\dagger)^{\otimes k}
\otimes
(U|\mathbf{0}\rangle\langle\mathbf{0}|U^\dagger)^{\otimes m}.
\end{equation}
Therefore, using the identity:
\[
\Tr[(X_1\otimes X_2)(Y_1\otimes Y_2)]
=
\Tr[X_1Y_1]\Tr[X_2Y_2],
\]
we obtain:
\begin{equation}
\gamma = \int \Tr [ (U|\mathbf{0}\rangle\langle\mathbf{0}|U^\dagger)^{\otimes k} (P_{\mathbb{1}}^{k,r})^\perp] \Tr [ (U|\mathbf{0}\rangle\langle\mathbf{0}|U^\dagger)^{\otimes m} P_{\mathbb{1}}^{m,0} ]dU.
\label{eq:gamma-factorized}
\end{equation}
Taking $P_{\mathbb{1}}^{m,0}=\nu_0^{\otimes m} = (|\mathbf{0}\rangle\langle\mathbf{0}|)^{\otimes m}$ with Equations \eqref{eq:qU-def} and \eqref{eq:qU-power-def}, we have:
\begin{equation}
\Tr\left[
(U|\mathbf{0}\rangle\langle\mathbf{0}|U^\dagger)^{\otimes m}
P_{\mathbb{1}}^{m,0}
\right]
=
q_U^m.
\end{equation}
Moreover, Definition \ref{def:iid} of $(P_{\mathbb{1}}^{k,r})^\perp$ lets us explicitly decompose the second factor as:
\begin{align}
\Tr\left[
(U|\mathbf{0}\rangle\langle\mathbf{0}|U^\dagger)^{\otimes k} (P_{\mathbb{1}}^{k,r})^\perp \right] &= \sum_{\substack{S\subseteq\{1,\ldots,k\} \\ |S|>r}} (1-q_U)^{|S|}q_U^{k-|S|} \\
&= \sum_{s=r+1}^{k} \binom{k}{s} (1-q_U)^s q_U^{k-s}.
\end{align}
Therefore, $\gamma$ can be expressed as:
\begin{equation}
\label{eq:gamma-binomial-integral}
\gamma = \int_{\Mcal_n} q_U^m \sum_{s=r+1}^{k} \binom{k}{s} (1-q_U)^s q_U^{k-s} \,dU.
\end{equation}

We now decompose $\gamma$ into an explicit sum over moments $q_U$. For \(0\leq q\leq1\), we have the following identity for the tail of a binomial distribution:
\begin{align}
\label{eq:binomial-tail}
\sum_{s=r+1}^{k} \binom{k}{s} (1-q)^s q^{k-s} &= \Pr[\operatorname{Bin}(k,1-q)\geq r+1] \\
&= (1 - q)^{r+1} \sum_{t=0}^{k-r-1} \binom{r+t}{r} q^t,
\end{align}
where the factor $(1 - q)^{r+1}$ can be expanded using the binomial theorem:
\begin{equation}
\label{eq:expand-last-factor}
(1-q)^{r+1} = \sum_{a=0}^{r+1} (-1)^a \binom{r+1}{a} q^a.
\end{equation}
Direct substitution into Equation \eqref{eq:gamma-binomial-integral} and evaluating the integral (Equation \eqref{eq:qU-moment}) gives:
\begin{align}
\gamma &= \sum_{t=0}^{k-r-1} \binom{r+t}{r} \int_{\Mcal_n} q_U^{m+t} (1 - q_U)^{r+1} dU \\
&= \sum_{t=0}^{k-r-1} \binom{r+t}{r} \sum_{a=0}^{r+1} (-1)^a \binom{r+1}{a} \int_{\Mcal_n} q_U^{a + m + t} dU \\
&= \sum_{t=0}^{k-r-1} \sum_{a=0}^{r+1} \binom{r+t}{r} \binom{r+1}{a} (-1)^a \frac{1}{D_{m+a+t}}. \label{eq:gamma-binomial-full}
\end{align}
Along with the definition for $\delta$, this gives Equation \eqref{eq:delta-full} in the main text. Next, we derive a simplified upper bound on $\delta$. If at least \(r+1\) failures occur, then there exists a subset of
\(r+1\) prescribed positions in which all entries are failures. By the
union bound,
\begin{equation}
\label{eq:binomial-union-bound}
\Pr[\operatorname{Bin}(k,1-q)\geq r+1]
\leq
\binom{k}{r+1}(1-q)^{r+1}.
\end{equation}
Applying this estimate, as well as the binomial expansion in Equation \eqref{eq:expand-last-factor} to \(q=q_U\) in Equation \eqref{eq:gamma-binomial-integral}, we obtain:
\begin{equation}
\label{eq:gamma-union-bound-integral}
\gamma \leq \binom{k}{r+1} \int_{\Mcal_n} q_U^m(1-q_U)^{r+1} \,dU.
\end{equation}
Then, applying the moment identity for the vacuum-overlap parameter, we get:
\begin{equation}
\label{eq:gamma-upper}
\gamma
\leq
\binom{k}{r+1}
\sum_{a=0}^{r+1}
(-1)^a
\binom{r+1}{a}
\frac{1}{D_{m+a}}.
\end{equation}
Combining this with Equation \eqref{eq:delta_raw} gives a bound for $\delta$:
\begin{equation}
\label{eq:delta-finite-difference}
\delta
\leq
3
\binom{k}{r+1}
D_m
\sum_{a=0}^{r+1}
(-1)^a
\binom{r+1}{a}
\frac{1}{D_{m+a}}.
\end{equation}

We now bound the sum in Equation \eqref{eq:delta-finite-difference}, giving a simplified form for $\delta$. For convenience, we set:
\begin{equation}
\label{eq:def_h}
h:=r+1,
\end{equation}
and
\begin{equation}
\label{eq:def_Ln}
L_n:=\frac{n(n-1)}{2}.
\end{equation}
We assume \(n\geq2\), so that \(L_n>0\). $L_n$ gives the number of factors in the product formula for $D_\ell$ (Equation \eqref{eq:Dell-product}). Write \(D_x\) as a polynomial in \(x\):
\begin{equation}
D_x = 2 \prod_{1\leq i<j\leq n} \frac{x+2n-i-j}{2n-i-j} = C_n \prod_{\alpha=1}^{L_n}(x+c_\alpha),
\end{equation}
where:
\begin{equation}
C_n := 2 \prod_{1\leq i<j\leq n} \frac{1}{2n-i-j},
\end{equation}
and where \(c_\alpha\) runs over the multiset:
\begin{equation}
\{c_\alpha\}_{\alpha=1}^{L_n}
    =
    \{2n-i-j:1\leq i<j\leq n\}.
\end{equation}
In particular, it satisfies $c_\alpha\geq1$. Next, define:
\begin{equation}
\label{eq:definition_fx}
f(x):=\frac{1}{D_x}
=
C_n^{-1}\prod_{\alpha=1}^{L_n}(x+c_\alpha)^{-1}.
\end{equation}
With the convention $\Delta f(m):=f(m+1)-f(m)$, we write the \(h\)-fold forward finite difference:
\begin{equation} \label{eq:finite-difference-definition}
    \Delta^h f(m)
    =
    \sum_{a=0}^{h}
    (-1)^{h-a}
    \binom{h}{a}
    f(m+a),
\end{equation}
which gives us:
\begin{equation}
    (-1)^h\Delta^h f(m)
    =
    \sum_{a=0}^{h}
    (-1)^a
    \binom{h}{a}
    f(m+a).
\end{equation}
From Equation \eqref{eq:definition_fx} we have \(f(m+a)=1/D_{m+a}\), hence:
\begin{equation}    
    (-1)^h\Delta^h f(m) = \sum_{a=0}^{h} (-1)^a  \binom{h}{a} \frac{1}{D_{m+a}}. \label{eq:finite-difference}
\end{equation}
The standard integral representation of finite differences gives:
\begin{equation}
\label{eq:standard_integral}
\Delta^h f(m)
=
\int_{[0,1]^h}
f^{(h)}(m+t_1+\cdots+t_h)\,dt_1\cdots dt_h.
\end{equation}
The above identity follows by iterating the fundamental theorem of calculus.
For \(h=1\),
\[
    \Delta f(m)=f(m+1)-f(m)
    =
    \int_0^1 f'(m+t_1)\,dt_1.
\]
Applying the same argument successively to \(\Delta f,\Delta^2 f,\ldots, \Delta^h f\) gives Equation \eqref{eq:standard_integral}. We now find the derivative $f^{(h)}(x)$, using the Leibniz rule on the product $\prod_\alpha(x+c_\alpha)^{-1}$:
\begin{equation}
f^{(h)}(x) = \frac{1}{C_n} \sum_{\substack{\beta_1,\ldots,\beta_{L_n}\geq0\\
\beta_1+\cdots+\beta_{L_n}=h}} \binom{h}{\beta_1, .. ,\beta_{L_n}} \prod_{\alpha=1}^{L_n} \frac{d^{\beta_\alpha}}{dx^{\beta_\alpha}} (x+c_\alpha)^{-1}.
\end{equation}
Each derivative in the product evaluates to:
\begin{equation}
    \frac{d^{\beta_\alpha}}{dx^{\beta_\alpha}}
    (x+c_\alpha)^{-1}
    =
    (-1)^{\beta_\alpha}\beta_\alpha!
    (x+c_\alpha)^{-1-\beta_\alpha}.
\end{equation}
Overall, the $\beta_\alpha!$ factors cancel against the denominator of the multinomial coefficient, and the signs combine to give $\prod_\alpha(-1)^{\beta_\alpha}=(-1)^h$. We obtain:
\begin{equation}
f^{(h)}(x) = (-1)^h C_n^{-1} h! \sum_{\substack{\beta_1,\ldots,\beta_{L_n}\geq0\\
\beta_1+\cdots+\beta_{L_n}=h}} \prod_{\alpha=1}^{L_n} (x+c_\alpha)^{-1-\beta_\alpha}.
\end{equation}
Substituting $f(x)$ from Equation \eqref{eq:definition_fx} back into this expression and rearranging yields:
\begin{equation}
(-1)^h f^{(h)}(x) = f(x)\, h! \sum_{\substack{\beta_1,\ldots,\beta_{L_n}\geq0\\ \beta_1+\cdots+\beta_{L_n}=h}} \prod_{\alpha=1}^{L_n} (x+c_\alpha)^{-\beta_\alpha}.
\end{equation}
Since $x+c_\alpha\geq x+1$, the right hand side is a completely monotone sum of positive terms. We bound each term as:
\begin{equation}
\prod_{\alpha=1}^{L_n} (x+c_\alpha)^{-\beta_\alpha} \leq \prod_{\alpha=1}^{L_n} (x+1)^{-\beta_\alpha} = (x+1)^{-h},
\end{equation}
and the number of terms in the weak composition of \(h\) into \(L_n\) parts, given by the binomial coefficient \(\binom{L_n+h-1}{h}\). Overall, we obtain:
\begin{align}
(-1)^h f^{(h)}(x)
&\leq
f(x)\,
h!
\binom{L_n+h-1}{h}
\frac{1}{(x+1)^h}
\\
&=
f(x)
\frac{(L_n)_h}{(x+1)^h},
\end{align}
where $(L_n)_h$ is the rising factorial:
\begin{equation}
(L_n)_h
:=
L_n(L_n+1)\cdots(L_n+h-1).
\end{equation}
Re-inserting this bound into the integral representation of finite differences (Equation \eqref{eq:standard_integral}) gives:
\begin{align}
(-1)^h\Delta^h f(m) \label{eq:finite-difference-bound} 
&= \int_{[0,1]^h} (-1)^h f^{(h)}(m+t_1+\cdots+t_h)\,dt_1\cdots dt_h \\
&\leq \int_{[0,1]^h} (L_n)_h \frac{ f(m+t_1+\cdots+t_h)} {(m+t_1+\cdots+t_h+1)^h} \,dt_1\cdots dt_h \\
&\leq \frac{(L_n)_h}{(m+1)^h}f(m),
\end{align}
where we have used the fact that \(f(x)\) is positive and decreasing in \(x \geq 0\), as well as that \( (m+t_1+\cdots+t_h+1) \geq (m+1) \).
Since $f(m)=1/D_m$, Equation \eqref{eq:finite-difference} gives:
\begin{equation}
\sum_{a=0}^{h} (-1)^a  \binom{h}{a} \frac{1}{D_{m+a}} = (-1)^h\Delta^h f(m) 
\leq \frac{(L_n)_h}{(m+1)^h} \frac{1}{D_m}.
\end{equation}
Therefore, we have the following bound:
\begin{equation}
\label{eq:bound_v0}
D_m
\sum_{a=0}^{h}
(-1)^a
\binom{h}{a}
\frac{1}{D_{m+a}}
\leq
\frac{(L_n)_h}{(m+1)^h}.
\end{equation}
Using Equation \eqref{eq:bound_v0}, and combining it with Equations \eqref{eq:delta-finite-difference} and~\eqref{eq:def_h}, we obtain:
\begin{equation}
\label{eq:delta-rising-factorial}
\delta \leq 3 \binom{k}{r+1} \frac{(L_n)_{r+1}}{(m+1)^{r+1}},
\end{equation}
which is the left-hand side of Equation \eqref{eq:delta-simple}. To obtain the right-hand side, the above can be expanded into a product form, giving:
\begin{equation} \label{eq:delta-product-form}
\delta \leq \frac{3}{(r+1)!} \prod_{j=0}^r \frac{(k-j)(L_n+j)}{(m+1)}.
\end{equation}
We apply the AM-GM inequality to each product separately, giving:
\begin{align}
\prod_{j=0}^r (k-j) &\leq \left( k - \frac{r}{2} \right)^{r+1}, & \prod_{j=0}^r (L_n+j) &\leq \left( L_n + \frac{r}{2} \right)^{r+1}.
\end{align}
Substituting the AM-GM bounds into Equation \eqref{eq:delta-product-form}, we obtain:
\begin{equation}
\label{eq:delta-lambda-bound}
\delta \leq \frac{3}{(r+1)!} \lambda_r^{r+1},
\end{equation}
with the parameter \(\lambda_r\) defined as:
\begin{equation}
    \lambda_r := \frac{\left( k - \frac{r}{2} \right)\left(L_n + \frac{r}{2} \right)}{\left(m+1\right)}.
\end{equation}
This reproduces equation \eqref{eq:delta-simple} in Theorem \ref{thm:fermionic-gaussian-definetti}. We can also obtain a simpler (albeit weaker) form of $\lambda$ independent of $r$. Since we have that:
\begin{align}
\left( k - \frac{r}{2} \right)\left(L_n + \frac{r}{2} \right) &\leq \left( \frac{(k - \frac{r}{2}) + (L_n + \frac{r}{2})}{2} \right)^2 \\
&=  \frac{(k + L_n)^2}{4},
\end{align}
we can write:
\begin{equation}
\lambda_r \leq \frac{(k + L_n)^2}{4(m+1)}.
\end{equation}
\end{proof}

For a simplified 'full' form of $\delta$ in Equation \eqref{eq:delta-full}, we can also return to the exact expression for $\gamma$ in Equation \eqref{eq:gamma-binomial-full}. We set $m + t$ as the argument of the $h$-fold finite difference in Equation \eqref{eq:finite-difference-definition} and re-insert $h = r+1$. The bound from Equation \eqref{eq:finite-difference-bound} then implies:
\begin{align}
    (-1)^h\Delta^{r+1} f(m + t) &= \sum_{a=0}^{r+1} (-1)^a  \binom{r+1}{a} \frac{1}{D_{m+a+ t}} \\
&\leq \frac{(L_n)_{r+1}}{(m+t+1)^{r+1}} \frac{1}{D_{m+t}} 
\end{align}
Direct substitution into Equation \eqref{eq:gamma-binomial-full} gives:
\begin{equation}
    \gamma  \leq \frac{ (L_n)_{r+1}}{(m+1)^{r+1}} \sum_{t=0}^{k-r-1} \binom{r+t}{r} \left( \frac{m + 1}{m + t + 1} \right)^{r+1} \frac{1}{D_{m+t}}
\end{equation}
Inserting this into the definition of $\delta$ in Equation \eqref{eq:delta_raw}:
\begin{equation}
    \delta  \leq \frac{3 (L_n)_{r+1}}{(m+1)^{r+1}} \sum_{t=0}^{k-r-1} \binom{r+t}{r} \left( \frac{m + 1}{m + t + 1} \right)^{r+1} \frac{D_m}{D_{m+t}}.
\end{equation}
Bounding $\left( \frac{m + 1}{m + t + 1} \right)^{r+1} \frac{D_m}{D_{m+t}} \leq 1$ and applying the hockey-stick identity also reproduces Equation \eqref{eq:delta-rising-factorial}. 

Most of the slack in the bounds of Theorem \ref{thm:fermionic-gaussian-definetti} comes from the union bound (Equation \eqref{eq:binomial-union-bound}). For example, if $r=0$, Equation \eqref{eq:delta-finite-difference} (obtained immediately following its application) with $r = 0$ gives:
\begin{align}
    \delta &\leq 3 \binom{k}{1} D_m \sum_{a=0}^{1} (-1)^a \binom{1}{a} \frac{1}{D_{m+a}}  \\
    &= 3 k \left( 1 - \frac{D_{m}}{D_{m+1}} \right) \\
    & \leq \frac{3 k L_n}{m+1},
\end{align}
in agreement with the further simplified Equation \eqref{eq:delta-lambda-bound}. On the other hand, expanding $\delta \leq 3 D_m \gamma$ via Equation \eqref{eq:gamma-binomial-integral} with $r=0$ gives:
\begin{align}
    \delta & \leq 3 D_m \int_{\Mcal_n} q_U^m \sum_{s=1}^{k} \binom{k}{s} (1-q_U)^s q_U^{k-s} \,dU \\
    &= 3 D_m \int_{\Mcal_n} (q_U^m - q_U^{m+k}) \,dU \\
    &= 3 \left( 1 - \frac{D_m}{D_{m+k}} \right) \\
     &\leq \frac{3 k L_n}{m+k+1},
\end{align}
in agreement up to a constant factor with the bound in \cite{Sierant_2026}, which is tighter for $r = 0$ than Equation \eqref{eq:delta-simple}. 

\

An alternative approach which avoids the union bound leads to Corollary \ref{cor:fermionic-gaussian-renner-bound}. As this requires a more in-depth proof, we have chosen to present it separately here. As before, we write the Gaussian-symmetric dimension formula $D_x$:
\begin{equation}
    D_x = 2 \prod_{1\leq i<j\leq n}  \frac{x+2n-i-j}{2n-i-j} =   2\prod_{\alpha=1}^{L_n}   \frac{x+c_\alpha}{c_\alpha},  \label{eq:renner-Dx}
\end{equation}
where $c_\alpha \geq 1$ again runs over the multiset: 
\begin{equation}
    \{c_\alpha\}_{\alpha=1}^{L_n}
    =
    \{2n-i-j:1\leq i<j\leq n\}.
    \label{eq:renner-calpha}
\end{equation}
We will admit noninteger values $x>-1$, using the positive polynomial continuation of $D_x$ in Equation \eqref{eq:renner-Dx}. Next, recall that $\delta\leq 3D_m\gamma$, where $\gamma$ was given an exact expression in Equation \eqref{eq:gamma-binomial-integral}:
\begin{equation}
    \gamma  =  \int_{\Mcal_n}  q_U^m  \sum_{s=h}^{k}  \binom{k}{s}  (1-q_U)^s q_U^{k-s}  \,dU,
    \label{eq:renner-gamma-binomial}
\end{equation}
where $q_U:=|\langle 0|U|0\rangle|^2$, $h:=r+1$. Also recall Equation \eqref{eq:binomial-tail}, expressing the sum in the integrand as a Binomial tail:
\begin{equation}
    \sum_{s=h}^{k}
    \binom{k}{s}(1-q)^s q^{k-s}
    =
    \Pr[\operatorname{Bin}(k,1-q)\geq h].
    \label{eq:renner-binomial-tail}
\end{equation}
With this in mind, we may absorb the factor $q_U^m$ in Equation \eqref{eq:renner-gamma-binomial} into a probability measure. For $m\geq 1$, define the $m$-tilted Haar measure:
\begin{equation}
d\nu_m(U):=D_m q_U^m\,dU.
\end{equation}
Indeed, by the moment identity in Equation \eqref{eq:qU-moment} we have $\int_{\Mcal_n} d\nu_m(U) = 1$. Taking:
\begin{equation}
q(U) := q_U = |\langle \mathbf{0} |U|\mathbf{0}\rangle|^2,
\end{equation}
The definition for $\delta$ combined with Equations \eqref{eq:renner-gamma-binomial} and \eqref{eq:renner-binomial-tail} gives:
\begin{equation}
    \delta \leq 3\, \mathbb{E}_{\nu_m}
    \left(  \Pr\!\left[   \operatorname{Bin}(k,1-Q)\geq h    \,\middle|\,Q  \right] \right).
    \label{eq:renner-delta-expectation}
\end{equation}
The next lemma identifies the distribution of \(Q\) under the tilted measure.  This identification is useful because it gives negative moments of \(Q\), whereas the original moment identity \eqref{eq:qU-moment} only states positive integer moments. We first recall the notation for the beta function and the beta
probability distribution. For $p,q>0$, the Euler beta function is defined by:
\begin{equation}
    \mathrm{B}(p,q)
    :=
    \int_0^1 x^{p-1}(1-x)^{q-1}\,dx
    =
    \frac{\Gamma(p)\Gamma(q)}{\Gamma(p+q)}.
    \label{eq:euler-beta-function}
\end{equation}
The notation: 
\begin{equation}
Z\sim\operatorname{Beta}(p,q)
\end{equation} 
means that $Z$ is a random variable supported on $(0,1)$, with
probability density:
\begin{equation}
    f_{p,q}(x)
    =
    \frac{x^{p-1}(1-x)^{q-1}}
         {\mathrm{B}(p,q)}
    \mathbb{1}_{(0,1)}(x).
    \label{eq:beta-distribution-density}
\end{equation}
Here, $\mathrm{B}(p,q)$ is the Euler beta function appearing as the
normalization constant, whereas $\operatorname{Beta}(p,q)$ denotes
the corresponding probability distribution. In the special case $q=1$, one has:
\begin{equation}
    \mathrm{B}(p,1)=\frac{1}{p},
\end{equation}
and therefore:
\begin{equation}
    Z\sim\operatorname{Beta}(p,1)
    \quad\Longleftrightarrow\quad
    f_{p,1}(x)
    =
    p x^{p-1}\mathbb{1}_{(0,1)}(x).
    \label{eq:beta-p-one-density}
\end{equation}

\begin{lemma}[Beta-product Representation]
\label{lem:renner-beta-product}
Under the probability measure $\nu_m$, the overlap random variable
$Q$ has the same probability distribution as a product:
\begin{equation}
    Q
    \overset{\mathrm{d}}{=}
    \prod_{j=1}^{L_n} Z_j,
    \label{eq:renner-beta-product}
\end{equation}
where $Z_1,\ldots,Z_{L_n}$ are independent random variables satisfying
\begin{equation}
    Z_j
    \sim
    \operatorname{Beta}(m+c_j,1),
    \qquad
    j=1,\ldots,L_n.
    \label{eq:renner-beta-parameters}
\end{equation}
Here, $\overset{\mathrm{d}}{=}$ denotes equality in probability
distribution. Consequently, for every real number $a$ such that:
\begin{equation}
    0\leq a<m+1,
    \label{eq:renner-a-domain}
\end{equation}
the negative moment of $Q$ is finite and is given by:
\begin{equation}
    \mathbb{E}_{\nu_m}\!\left[Q^{-a}\right]
    =
    \prod_{j=1}^{L_n}
    \frac{m+c_j}{m+c_j-a}
    =
    \frac{D_m}{D_{m-a}}.
    \label{eq:renner-negative-moment}
\end{equation}
\end{lemma}

\begin{proof}
For every integer $s\geq0$, the definition of the tilted measure $\nu_m$ and the overlap-moment identity give:
\begin{align}
    \mathbb{E}_{\nu_m}[Q^s]
    &=
    D_m\int_{M_n}q_U^{m+s}\,dU
    \notag\\
    &=
    \frac{D_m}{D_{m+s}}
    \notag\\
    &=
    \prod_{j=1}^{L_n}
    \frac{m+c_j}{m+c_j+s}.
    \label{eq:Q-positive-moments}
\end{align}
On the other hand, since:
\[
    Z_j\sim\operatorname{Beta}(m+c_j,1),
\]
its probability density is:
\begin{equation}
    f_j(x)
    =
    (m+c_j)x^{m+c_j-1}\mathbb{1}_{(0,1)}(x).
\end{equation}
Therefore,
\begin{align}
    \mathbb{E}[Z_j^s]
    =
    (m+c_j)
    \int_0^1 x^{m+c_j+s-1}\,dx
    =
    \frac{m+c_j}{m+c_j+s}.
    \label{eq:Zj-positive-moment}
\end{align}
Using independence, we obtain
\begin{align}
    \mathbb{E}\left[
        \left(\prod_{j=1}^{L_n}Z_j\right)^s
    \right]
    &=
    \prod_{j=1}^{L_n}\mathbb{E}[Z_j^s] \\
   & =
    \prod_{j=1}^{L_n}
    \frac{m+c_j}{m+c_j+s}
    =
    \mathbb{E}_{\nu_m}[Q^s].
    \label{eq:matching-moments}
\end{align}
Both $Q$ and $\prod_{j=1}^{L_n}Z_j$ are supported on $[0,1]$.
Since the Hausdorff moment problem on $[0,1]$ is determinate, equality
of all nonnegative integer moments implies:
\[
    Q
    \overset{\mathrm{d}}{=}
    \prod_{j=1}^{L_n}Z_j.
\]
It remains to evaluate the negative moments. For any real
$a<m+c_j$,
\begin{align}
    \mathbb{E}[Z_j^{-a}]
    =
    (m+c_j)
    \int_0^1 x^{m+c_j-a-1}\,dx
    =
    \frac{m+c_j}{m+c_j-a}.
    \label{eq:Zj-negative-moment}
\end{align}
Since $\min_j c_j=1$, all factors are finite whenever
$a<m+1$. Independence therefore gives:
\begin{align}
    \mathbb{E}_{\nu_m}[Q^{-a}]
    &=
    \prod_{j=1}^{L_n}
    \mathbb{E}[Z_j^{-a}]
    \notag\\
    &=
    \prod_{j=1}^{L_n}
    \frac{m+c_j}{m+c_j-a}
    \notag\\
    &=
    \frac{D_m}{D_{m-a}},
\end{align}
where the last equality follows from the polynomial continuation for $D_x$ in Equation \eqref{eq:renner-Dx}.
\end{proof}

Next, we will produce a Chernoff bound for the conditional binomial tail. Recall that, under the tilted probability measure $\nu_m$, the vacuum-overlap parameter $Q = q_U = |\langle \mathbf{0} |U|\mathbf{0}\rangle|^2$ is a random variable taking values in $[0,1]$.  To interpret the binomial tail probabilistically, one may introduce an auxiliary random variable $X$ as follows.  First sample $U$ according to $\nu_m$, and hence determine $Q=q_U$.  Then conditionally on $Q=q$, sample:
\begin{equation}
    X\mid (Q=q )
    \sim
    \operatorname{Bin}(k,1-q).
    \label{eq:conditional-binomial-law}
\end{equation}
Thus, $X$ counts the number of failures among the $k$ retained replicas when the overlap parameter is fixed to be $q$. With this notation, the exact binomial-tail representation obtained above reads:
\begin{equation}
    \delta
    \leq
    3\,
    \mathbb{E}_{\nu_m}
    \left(
        \Pr\!\left[
            X\geq h\,\middle|\,Q
        \right]
    \right),
    \qquad
    h:=r+1.
    \label{eq:delta-conditional-representation}
\end{equation}
We now provide an estimate for $\Pr\!\left[   \operatorname{Bin}(k,1-Q)\geq h    \,\middle|\,Q  \right]$, first pointwise for a fixed value of $Q=q$, and then averaged over the distribution of $Q$. 

\begin{lemma}[Conditional Exponential Estimate]
\label{lem:renner-conditional-chernoff}
Let $1\leq h\leq k$, let $q\in(0,1]$, and let $a\geq0$. Then
\begin{equation}
    \Pr\!\left[
        \operatorname{Bin}(k,1-q)\geq h
    \right]
    \leq
    \left(1+\frac{a}{k}\right)^{-h}q^{-a}.
    \label{eq:renner-conditional-bound}
\end{equation}
\end{lemma}
\begin{proof}
Fix $q\in(0,1]$.  Let
$Y_1^{(q)},\ldots,Y_k^{(q)}$ be independent Bernoulli random
variables satisfying:
\begin{equation}
    \Pr\!\left[Y_i^{(q)}=1\right]=1-q,
    \qquad
    \Pr\!\left[Y_i^{(q)}=0\right]=q.
    \label{eq:bernoulli-failure-variable}
\end{equation}
The value $Y_i^{(q)}=1$ represents a failure on the $i$-th retained
replica.  Define the total number of failures $X_q$:
\begin{equation}
    X_q:=\sum_{i=1}^{k}Y_i^{(q)}.
    \label{eq:number-of-failures}
\end{equation}
By construction, $X_q\sim\operatorname{Bin}(k,1-q)$. We first introduce a standard Chernoff parameter $\theta\geq0$.
Since the function $x\mapsto e^{\theta x}$ is increasing, the events
$\{X_q\geq h\}$ and
$\{e^{\theta X_q}\geq e^{\theta h}\}$ coincide. Therefore,
\begin{equation}
    \Pr[ X_q\geq h ]
    =
    \Pr\!\left[
        e^{\theta X_q}\geq e^{\theta h}
    \right].
    \label{eq:exponential-event}
\end{equation}
The random variable $e^{\theta X_q}$ is nonnegative, so Markov's
inequality gives:
\begin{align}
    \Pr[ X_q\geq h ]
    &\leq
    e^{-\theta h}
    \mathbb{E}\!\left[e^{\theta X_q}\right].
    \label{eq:markov-chernoff}
\end{align}
We now evaluate the moment-generating function of $X_q$.  Using Equation
\eqref{eq:number-of-failures} and the independence of the variables
$Y_i^{(q)}$, we obtain
\begin{align}
    \mathbb{E}\!\left[e^{\theta X_q}\right]
    =
    \mathbb{E}\!\left[
        e^{\theta\sum_{i=1}^{k}Y_i^{(q)}}
    \right]
    =
    \mathbb{E}\!\left[
        \prod_{i=1}^{k}e^{\theta Y_i^{(q)}}
    \right]
    =
    \prod_{i=1}^{k}
    \mathbb{E}\!\left[e^{\theta Y_i^{(q)}}\right].
    \label{eq:mgf-factorization}
\end{align}
For each Bernoulli variable:
\begin{align}
    \mathbb{E}\!\left[e^{\theta Y_i^{(q)}}\right]
    =
    q e^{\theta\cdot0}
    +(1-q)e^{\theta\cdot1}
    =
    q+(1-q)e^\theta.
    \label{eq:bernoulli-mgf}
\end{align}
Consequently,
\begin{equation}
    \mathbb{E}\!\left[e^{\theta X_q}\right]
    =
    \bigl(q+(1-q)e^\theta\bigr)^k.
    \label{eq:binomial-mgf}
\end{equation}
Substituting this identity into Equation \eqref{eq:markov-chernoff} gives
the standard Chernoff estimate:
\begin{equation}
    \Pr[ X_q\geq h ]
    \leq
    e^{-\theta h}
    \bigl(q+(1-q)e^\theta\bigr)^k.
    \label{eq:standard-chernoff}
\end{equation}
We now reparametrize $\theta$ by introducing an auxiliary real
parameter $a\geq0$ through $a:=k(e^\theta-1)$. Equivalently:
\begin{equation}
    e^\theta=1+\frac{a}{k},
    \qquad
    \theta=\log\left(1+\frac{a}{k}\right).
    \label{eq:renner-theta-a}
\end{equation}
The parameter $a$ is only a Chernoff optimization parameter; it is
not a physical parameter and it is not a replica index.  Since $a\geq0$, the corresponding value of $\theta$ is nonnegative. Using $e^\theta=1+a/k$, the moment-generating factor becomes:
\begin{align}
    q+(1-q)e^\theta
    &=
    q+(1-q)\left(1+\frac{a}{k}\right)
    \notag\\
    &=
    q+(1-q)+\frac{a}{k}(1-q)
    \notag\\
    &=
    1+\frac{a}{k}(1-q).
    \label{eq:renner-mgf-rewrite}
\end{align}
Hence,
\begin{equation}
    \bigl(q+(1-q)e^\theta\bigr)^k
    =
    \left[
        1+\frac{a}{k}(1-q)
    \right]^k.
    \label{eq:mgf-after-reparametrization}
\end{equation}
We next use the elementary inequality:
\begin{equation}
    1+x\leq e^x,
    \qquad x\geq0.
    \label{eq:one-plus-x}
\end{equation}
Taking $x$ in Equation \eqref{eq:one-plus-x} to be:
\begin{equation}
    x=\frac{a}{k}(1-q)\geq0
\end{equation}
gives:
\begin{align}
    \left[
        1+\frac{a}{k}(1-q)
    \right]^k
    &\leq
    \left[
        \exp\left(
            \frac{a}{k}(1-q)
        \right)
    \right]^k
    \notag\\
    &=
    e^{a(1-q)}.
    \label{eq:mgf-first-relaxation}
\end{align}
For $q\in(0,1]$, one also has:
\begin{equation}
    1-q\leq-\log q.
    \label{eq:one-minus-q-log}
\end{equation}
Indeed,
\begin{equation}
    -\log q
    =
    \int_q^1\frac{dt}{t}
    \geq
    \int_q^1dt
    =
    1-q.
    \label{eq:one-minus-q-log-proof}
\end{equation}
Since $a\geq0$, multiplication by $a$ and exponentiation yields:
\begin{equation}
    e^{a(1-q)}
    \leq
    e^{-a\log q}
    =
    q^{-a}.
    \label{eq:renner-q-minus-a}
\end{equation}
Combining Equations
\eqref{eq:mgf-after-reparametrization},
\eqref{eq:mgf-first-relaxation}, and
\eqref{eq:renner-q-minus-a}, we obtain:
\begin{equation}
    \bigl(q+(1-q)e^\theta\bigr)^k
    \leq
    q^{-a}.
    \label{eq:mgf-final-bound}
\end{equation}
Finally, the first factor in Equation \eqref{eq:standard-chernoff} satisfies:
\begin{align}
    e^{-\theta h}
    &=
    \exp\left[
        -h\log\left(1+\frac{a}{k}\right)
    \right]
    \notag\\
    &=
    \left(1+\frac{a}{k}\right)^{-h}.
    \label{eq:renner-theta-prefactor}
\end{align}
Substitution of Equations \eqref{eq:mgf-final-bound} and
\eqref{eq:renner-theta-prefactor} into
\eqref{eq:standard-chernoff} gives:
\begin{equation}
    \Pr[ X_q\geq h ]
    \leq
    \left(1+\frac{a}{k}\right)^{-h}q^{-a},
\end{equation}
which proves Equation \eqref{eq:renner-conditional-bound}.
\end{proof}

\begin{remark} The restriction to $q>0$ is harmless in the application to the tilted
measure.  Indeed,
\begin{equation}
    d\nu_m(U)=D_mq_U^m\,dU,
\end{equation}
and, because $m\geq1$,
\begin{align}
    \nu_m\{U:q_U=0\}
    =
    D_m\int_{\{q_U=0\}}q_U^m\,dU
    =
    0.
\end{align}
Thus, $Q>0$ holds $\nu_m$-almost surely. 
\end{remark}
We now average the pointwise estimate from
Lemma \ref{lem:renner-conditional-chernoff} over the tilted distribution of the overlap variable $Q$.
\begin{theorem}[One-parameter Exponential Bound]
\label{thm:renner-master}
Let $m,k\geq1$ and $1\leq h\leq k$, where $r = h-1$. Then, for every real number $a$ satisfying:
\begin{equation}
    0\leq a<m+1,
    \label{eq:a-admissible-master}
\end{equation}
one has
\begin{equation}
    \delta
    \leq
    3
    \left(1+\frac{a}{k}\right)^{-h}
    \frac{D_m}{D_{m-a}}.
    \label{eq:renner-one-parameter}
\end{equation}
Here, for noninteger arguments, $D_x$ denotes the polynomial continuation:
\begin{equation}
    D_x
    =
    2\prod_{j=1}^{L_n}\frac{x+c_j}{c_j},
    \qquad
    x>-1.
    \label{eq:Dx-polynomial-continuation-master}
\end{equation}
Consequently $\delta$ in Corollary \ref{cor:fermionic-gaussian-renner-bound} satisfies:
\begin{equation}
    \delta
    \leq
    3\inf_{0\leq a<m+1}
    \left(1+\frac{a}{k}\right)^{-(r+1)}
    \frac{D_m}{D_{m-a}}
    .
    \label{eq:renner-master}
\end{equation}
Equivalently,
\begin{equation}
    \delta
    \leq
    3\inf_{0\leq a<m+1}
    \exp  \bigg[
        -(r+1)\log \left( 1+\frac{a}{k}\right)   + 
         \sum_{j=1}^{L_n}
        \log\frac{m+c_j}{m+c_j-a}
    \bigg]. \label{eq:renner-master-log}
\end{equation}
\end{theorem}

\begin{proof}
Recall from Equation \eqref{eq:delta-conditional-representation} that:
\begin{equation}
    \delta
    \leq
    3\,
    \mathbb{E}_{\nu_m}
    \left(
        \Pr\!\left[
            X\geq h\,\middle|\,Q
        \right]
    \right),
    \label{eq:delta-conditional-recalled}
\end{equation}
where, conditionally on $Q=q$,
\begin{equation}
    X\mid Q=q
    \sim
    \operatorname{Bin}(k,1-q).
    \label{eq:X-given-Q}
\end{equation}
Fix an admissible value $a\in[0,m+1)$.  For every $q\in(0,1]$,
Lemma~\ref{lem:renner-conditional-chernoff} gives:
\begin{equation}
    \Pr\!\left[
        X\geq h\,\middle|\,Q=q
    \right]
    \leq
    \left(1+\frac{a}{k}\right)^{-h}q^{-a}.
    \label{eq:conditional-bound-at-q}
\end{equation}
Because $Q>0$ holds $\nu_m$-almost surely, we may substitute the
random value $q=Q$ into this pointwise inequality.  We obtain:
\begin{equation}
    \Pr\!\left[
        X\geq h\,\middle|\,Q
    \right]
    \leq
    \left(1+\frac{a}{k}\right)^{-h}Q^{-a},
    \quad
    \nu_m\text{-almost surely}.
    \label{eq:conditional-bound-random-Q}
\end{equation}
Taking expectations and using Equation
\eqref{eq:delta-conditional-recalled}, we find:
\begin{align}
    \delta
    &\leq
    3\,
    \mathbb{E}_{\nu_m}
    \left[
        \left(1+\frac{a}{k}\right)^{-h}Q^{-a}
    \right]
    \notag\\
    &=
    3
    \left(1+\frac{a}{k}\right)^{-h}
    \mathbb{E}_{\nu_m}\!\left[Q^{-a}\right].
    \label{eq:average-conditional-bound}
\end{align}
The factor
\[
    \left(1+\frac{a}{k}\right)^{-h}
\]
is deterministic, and has therefore been taken outside the expectation. The restriction $a<m+1$ is precisely the condition under which the
negative moment of $Q$ is finite.  By
Lemma~\ref{lem:renner-beta-product},
\begin{equation}
    \mathbb{E}_{\nu_m}\!\left[Q^{-a}\right]
    =
    \frac{D_m}{D_{m-a}}.
    \label{eq:negative-moment-in-master-proof}
\end{equation}
Substituting this identity into
\eqref{eq:average-conditional-bound} proves
\eqref{eq:renner-one-parameter}.

Since \eqref{eq:renner-one-parameter} is valid for every
$a\in[0,m+1)$, we may take the infimum over all admissible values of
$a$.  This gives \eqref{eq:renner-master}. Finally, using the polynomial continuation (Equation \eqref{eq:renner-Dx} or \eqref{eq:Dx-polynomial-continuation-master}) we have:
\begin{align}
    \frac{D_m}{D_{m-a}}
    &=
    \frac{
        2\prod_{j=1}^{L_n}(m+c_j)/c_j
    }{
        2\prod_{j=1}^{L_n}(m-a+c_j)/c_j
    }
    \notag\\
    &=
    \prod_{j=1}^{L_n}
    \frac{m+c_j}{m+c_j-a}.
    \label{eq:dimension-ratio-product}
\end{align}
Therefore,
\begin{equation}
    \left(1+\frac{a}{k}\right)^{-(r+1)}
    \frac{D_m}{D_{m-a}} 
    =
    \exp\left[
        -(r+1)\log\left(1+\frac{a}{k}\right)
        +
        \sum_{j=1}^{L_n}
        \log\frac{m+c_j}{m+c_j-a}
    \right],
\end{equation}
which proves \eqref{eq:renner-master-log}.
\end{proof}
\begin{remark}
The parameter $a$ is auxiliary and should not be confused with the
physical replica numbers $m$ and $k$.  Increasing $a$ improves the
exponential factor:
\begin{equation}
    \left(1+\frac{a}{k}\right)^{-h},
\end{equation}
but simultaneously increases the negative-moment penalty:
\begin{equation}
    \frac{D_m}{D_{m-a}}.
\end{equation}
The infimum in \eqref{eq:renner-master} balances these two competing
effects.
\end{remark}
The optimised expression in Equation \eqref{eq:renner-master} is the
strongest bound produced by the preceding arguments.  For comparison
with the usual exponential de Finetti theorem, it is useful to make
the explicit choice:
\begin{equation}
    a=m.
    \label{eq:renner-choice-a-m}
\end{equation}
This choice is admissible because:
\begin{equation}
    0\leq m<m+1.
\end{equation}
We do not assert that $a=m$ is the optimal value of the Chernoff parameter.  Its advantage is that it produces a simple closed-form bound. Using the polynomial continuation in Equation
\eqref{eq:Dx-polynomial-continuation-master}, we obtain:
\begin{equation}
    D_0
    =
    2\prod_{j=1}^{L_n}\frac{c_j}{c_j}
    =
    2.
    \label{eq:D-zero-polynomial}
\end{equation}
Consequently,
\begin{equation}
    \frac{D_m}{D_{m-a}}
    \bigg|_{a=m}
    =
    \frac{D_m}{D_0}
    =
    \frac{D_m}{2}.
    \label{eq:renner-ratio-at-m}
\end{equation}
The value $D_0=2$ in \eqref{eq:D-zero-polynomial} is simply the value at $x=0$ of the polynomial continuation of $D_x$. We do not interpret this as the dimension of a physical zero-replica space.

Using Theorem \ref{thm:renner-master}, we are ready to prove Corollary \ref{cor:fermionic-gaussian-renner-bound} in the main text.
\begin{proof}[Proof of Corollary \ref{cor:fermionic-gaussian-renner-bound}]
Set $a=m$ in Equation \eqref{eq:renner-one-parameter}.  Then,
\begin{align}
    \delta
    &\leq
    3
    \left(1+\frac{m}{k}\right)^{-(r+1)}
    \frac{D_m}{D_0}
    \notag\\
    &=
    \frac{3D_m}{2}
    \left(1+\frac{m}{k}\right)^{-(r+1)}.
    \label{eq:corollary-substitution}
\end{align}
Since we have that:
\begin{equation}
    1+\frac{m}{k}
    =
    \frac{m+k}{k},
\end{equation}
we have:
\begin{equation}
    \left(1+\frac{m}{k}\right)^{-(r+1)}
    =
    \left(\frac{k}{m+k}\right)^{r+1}.
\end{equation}
This proves Equation \eqref{eq:renner-exact-power} of Corollary \ref{cor:fermionic-gaussian-renner-bound}.

We next relax the exact exponential rate in order to obtain the same
functional form as the usual Renner bound. For every $x\geq0$ we have that:
\begin{equation}
    \log(1+x)\geq\frac{x}{1+x}.
    \label{eq:renner-log-inequality}
\end{equation}
We prove Equation \eqref{eq:renner-log-inequality} for completeness. Define:
\begin{equation}
    g(x)
    :=
    \log(1+x)-\frac{x}{1+x}.
\end{equation}
Then, $g(0) = 0$ and:
\begin{align}
    g'(x)
    =
    \frac{1}{1+x}
    -
    \frac{1}{(1+x)^2}
    =
    \frac{x}{(1+x)^2}
    \geq0.
\end{align}
Thus $g$ is nondecreasing on $[0,\infty)$, and therefore
$g(x)\geq g(0)=0$.  This proves
\eqref{eq:renner-log-inequality}. Taking $x=m/k$ then gives:
\begin{align}
    \log\left(1+\frac{m}{k}\right)
    &\geq
    \frac{m/k}{1+m/k}
    \notag\\
    &=
    \frac{m}{m+k}.
    \label{eq:renner-rate-relaxation}
\end{align}
Multiplication by the negative number $-(r+1)$ reverses the
inequality:
\begin{equation}
    -(r+1)\log\left(1+\frac{m}{k}\right)
    \leq
    -\frac{m}{m+k}(r+1).
    \label{eq:negative-rate-relaxation}
\end{equation}
Consequently, Equation \eqref{eq:renner-exact-power} implies:
\begin{equation}
    \delta
    \leq
    3\exp\left[
        -\frac{m}{m+k}(r+1)
        +
        \log\frac{D_m}{2}
    \right].
    \label{eq:renner-exact-dimension-rate}
\end{equation}
Finally, we relax the exact dimension penalty.  Since $c_j\geq1$, we have:
\begin{equation}
    \frac{m}{c_j}\leq m,
\end{equation}
and hence:
\begin{equation}
    1+\frac{m}{c_j}\leq m+1.
\end{equation}
Using the product formula for $D_m$, we obtain
\begin{align}
    \frac{D_m}{2}
    &=
    \prod_{j=1}^{L_n}
    \frac{m+c_j}{c_j}
    \notag\\
    &=
    \prod_{j=1}^{L_n}
    \left(1+\frac{m}{c_j}\right)
    \notag\\
    &\leq
    \prod_{j=1}^{L_n}(m+1)
    \notag\\
    &=
    (m+1)^{L_n}.
    \label{eq:renner-dimension-relaxation}
\end{align}
Taking logarithms gives
\begin{equation}
    \log\frac{D_m}{2}
    \leq
    L_n\log(m+1).
    \label{eq:log-dimension-relaxation}
\end{equation}
Combining
\eqref{eq:renner-exact-dimension-rate} and
\eqref{eq:log-dimension-relaxation}, we arrive at the simplest direct
analogue of the bound in \cite{Renner_2007}:
\begin{equation}
    \delta
    \leq
    3\exp\left[
        -\frac{m}{m+k}(r+1)
        +
        L_n\log(m+1)
    \right],
    \label{eq:renner-simple}
\end{equation}
which reproduces Equation \eqref{eq:delta-renner-style} in the main text, and completes the proof.
\end{proof}
\section{Proof of Theorem \ref{thm:gaussian-invariant}}

In this section, we prove Theorem \ref{thm:gaussian-invariant}, which identifies partial traces of Gaussian-symmetric states with Gaussian-invariant states. We prove the two directions separately: first, that partial traces of Gaussian-symmetric states are Gaussian-invariant, and second, that every Gaussian-invariant state admits a Gaussian-symmetric purification.

We will make use of the following two lemmas:

\begin{lemma} \label{lem:partial-trace-cyclicity}
Let $\mathcal{A} \otimes \mathcal{B}$ be a bipartite Hilbert space, with $R$ supported on $\mathcal{B}$ and $X$ supported on $\mathcal{A} \otimes \mathcal{B}$. Then, the following identity holds:
\begin{equation}
\Tr_{\mathcal{B}}((\mathbb{1}_\mathcal{A} \otimes R_\mathcal{B}) X) = \Tr_{\mathcal{B}}(X (\mathbb{1}_\mathcal{A} \otimes R_\mathcal{B})).
\end{equation} 
\end{lemma}
\begin{proof}
Write out the partial trace explicitly, in terms of a basis $\{|b\rangle\}$ for $\mathcal{B}$:
\begin{align}
\Tr_{\mathcal{B}}((\mathbb{1}_\mathcal{A} \otimes R_\mathcal{B}) X) &= \sum_b (\mathbb{1}_\mathcal{A} \otimes \langle b|) (\mathbb{1}_\mathcal{A} \otimes R_\mathcal{B}) X (\mathbb{1}_\mathcal{A} \otimes |b\rangle) \\
&= \sum_b (\mathbb{1}_\mathcal{A} \otimes \langle b| R_\mathcal{B} ) X(\mathbb{1}_\mathcal{A} \otimes |b\rangle) \\
&= \sum_{b, b'} [R_\mathcal{B}]_{b b'} (\mathbb{1}_\mathcal{A} \otimes \langle b'|) X(\mathbb{1}_\mathcal{A} \otimes |b\rangle) \\
&= \sum_{b'} (\mathbb{1}_\mathcal{A} \otimes \langle b'|) X (\mathbb{1}_\mathcal{A} \otimes R_\mathcal{B}|b'\rangle) \\
&= \Tr_{\mathcal{B}}(X (\mathbb{1}_\mathcal{A} \otimes R_\mathcal{B})).
\end{align}
\end{proof}

Given that a state $\rho$ satisfies $[\rho, \Lambda_{ab}] = 0$, commutation with pair-parity operators $Q_{ab}$ and permutations $\Pi \in S_k$ automatically holds by the following lemma:

\begin{lemma} \label{lem:parity-permutation-commutation}
If a density operator $\rho$, supported on $\Hcal^{\otimes k}$, satisfies $[\rho, \Lambda_{ab}] = 0$ for all $1 \leq a < b \leq k$, then it also satisfies $[\rho, Q_{ab}] = 0$ and $[\rho, \Pi] = 0$ for all $1 \leq a < b \leq k$ and $\Pi \in S_k$.
\end{lemma}
\begin{proof}
    To show both claims, we express $Q_{ab}$ and $\Pi$ in terms of the bridge operators $\Lambda_{ab}$. First, we show that $Q_{ab} = e^{i \pi \Lambda_{ab} / 2}$. Since all of its terms $\gamma^{(a)}_\mu \gamma^{(b)}_\mu$ commute with each other, we can write:
    \begin{align}
        e^{\frac{i \pi}{2} \Lambda_{ab}} &= e^{\frac{i\pi}{2}\sum_{\mu=1}^{2n} \gamma^{(a)}_\mu \gamma^{(b)}_\mu} \\
        &= \prod_{\mu=1}^{2n} e^{\frac{i\pi}{2} \gamma^{(a)}_\mu \gamma^{(b)}_\mu} \\
        &= \prod_{\mu=1}^{2n} (i \gamma^{(a) }_\mu \gamma^{(b)}_\mu) \\
        &= \bigl( (-i)^n \prod_{\mu=1}^{2n} \gamma^{(a) }_\mu \bigr) \bigl( (-i)^n \prod_{\mu=1}^{2n} \gamma^{(b) }_\mu \bigr) = Q_{ab}.
    \end{align}
    Next, we express permutations $\Pi \in S_k$ in terms of bridge operators. Since transpositions generate the symmetric group $S_k$, it suffices to show that any transposition $\Pi_{ab}$ can be expressed in terms of bridge operators. For any complete basis of $d^2$ unitaries acting on $\Hcal$ orthonormal under the Hilbert--Schmidt inner product $d^{-1} \Tr(A^\dagger B)$, we may write the transposition $\Pi_{ab}$ as:
    \begin{equation}
        \Pi_{ab} = \frac{1}{d} \sum_A A^{(a)} (A^\dagger)^{(b)}.
    \end{equation}
    For ordered subsets $S = \{\mu_1, \mu_2, \dots, \mu_{r}\} \subseteq [2n]$, where $1 \leq \mu_1 < \dots < \mu_{r} \leq 2n$ and $r = |S|$, we may construct such a basis $\{ \gamma^{(a)}_S \}$ for $\Hcal$, where:
    \begin{align}
        \gamma^{(a)}_S &= \gamma^{(a)}_{\mu_1} \gamma^{(a)}_{\mu_2} \cdots \gamma^{(a)}_{\mu_{r}}, & (\gamma^{(a)}_S)^\dagger &= (-1)^{\frac{1}{2}r(r-1)} \gamma^{(a)}_S.
    \end{align}
    It follows that $\Pi_{ab}$ can be expressed as a linear combination of products of Majorana operators:
    \begin{align} 
        \Pi_{ab} &= \frac{1}{2^n} \sum_S (-1)^{\frac{1}{2}|S|(|S|-1)} \gamma^{(a)}_S \gamma^{(b)}_S \\
        &= \frac{1}{2^n} \sum_{r=0}^{2n} (-1)^{\frac{1}{2}r(r-1)} \sum_{|S|=r} \gamma^{(a)}_S \gamma^{(b)}_S. \label{eq:transposition-majorana-sum}
    \end{align}
    Using the elementary identity:
    \begin{equation}
    (-1)^{\frac{1}{2}r(r-1)} = \frac{1-i}{2}i^r + \frac{1+i}{2}(-i)^r,
    \end{equation}
    Equation \eqref{eq:transposition-majorana-sum} can be rewritten as:
    \begin{align}
        \Pi_{ab} &= \frac{1-i}{2^{n+1}} \sum_{r=0}^{2n} i^r \sum_{|S|=r} \gamma^{(a)}_S \gamma^{(b)}_S  \\
        & + \frac{1+i}{2^{n+1}} \sum_{r=0}^{2n} (-i)^r \sum_{|S|=r} \gamma^{(a)}_S \gamma^{(b)}_S.
    \end{align}
    Now, write out $e^{i \pi \Lambda_{ab} / 4}$ explicitly, again using the fact that all terms $\gamma^{(a)}_\mu \gamma^{(b)}_\mu$ in the sum commute:
    \begin{align}
        e^{\frac{i \pi}{4} \Lambda_{ab}} &= \prod_{\mu=1}^{2n} e^{\frac{i \pi}{4} \gamma^{(a)}_\mu \gamma^{(b)}_\mu} \\
        &= \prod_{\mu=1}^{2n} \frac{1}{\sqrt{2}} (\mathbb{1} + i \gamma^{(a)}_\mu \gamma^{(b)}_\mu) \\
        &= \frac{1}{2^n} \sum_{r=0}^{2n} i^r \sum_{|S|=r} \gamma^{(a)} _S \gamma^{(b)}_S.
    \end{align}
    It follows that:
    \begin{align}
        \Pi_{ab} &= \frac{1-i}{2} e^{\frac{i \pi}{4} \Lambda_{ab}} + \frac{1+i}{2} e^{-\frac{i \pi}{4} \Lambda_{ab}} \\
        & = \cos\left( \frac{\pi}{4} \Lambda_{ab} \right) + \sin\left( \frac{\pi}{4} \Lambda_{ab} \right).
    \end{align}
    Thus, $\Pi_{ab}$ is a function of the bridge operator $\Lambda_{ab}$. Consequently, if $[\rho, \Lambda_{ab}] = 0$, then $[\rho, \Pi_{ab}] = 0$. Since transpositions generate $S_k$, commutation with every bridge operator implies that $[\rho, \Pi] = 0$ for every $\Pi \in S_k$.  
\end{proof}

Now we proceed to prove Theorem \ref{thm:gaussian-invariant}.

\begin{proof}[Proof of Theorem \ref{thm:gaussian-invariant}] 
First, we show that $(2) \Rightarrow (1)$. We assume that $\bar{\rho}$ is a density operator supported on $\GSym^k(\Hcal \otimes \Kcal)$, so that $\bar{\Lambda}_{ab} \bar{\rho} = 0$ for all $1 \leq a < b \leq k$. Since $\bar{\rho}$ and $\bar{\Lambda}_{ab}$ are both Hermitian, it further follows that $\bar{\rho} \bar{\Lambda}_{ab} = 0$. 

Without loss of generality, we will use the Jordan--Wigner representation of the Majorana operators $\bar{\gamma}^{(a)}_\mu$ on the $a$-th replica of $\Hcal \otimes \Kcal$. This space contains $2n$ qubits, so that there are $4n$ Majorana operators in total, plus the parity operator: 
\begin{equation}\bar{\Gamma}^{(a)} = (-i)^{2n} \prod_{\mu = 1}^{4n} \bar{\gamma}^{(a)}_\mu = \prod_{j=1}^{2n} Z^{(a)}_j.
\end{equation}
We begin by writing the parity operator $\bar{\Gamma}^{(a)}$ in terms of separate parity operators $\Gamma^{(a)}_\Hcal$ and $\Gamma^{(a)}_\Kcal$ acting on the $\Hcal$ and $\Kcal$ subsystems of the $a$-th replica, respectively:
\begin{equation}
\bar{\Gamma}^{(a)} = \Gamma^{(a)}_\Hcal \otimes \Gamma^{(a)}_\Kcal 
= \left( \prod_{j=1}^{n} Z^{(a)}_j \right) \otimes \left( \prod_{j=n+1}^{2n} Z^{(a)}_j \right).
\end{equation}
This lets us write the Majorana operators $\bar{\gamma}^{(a)}_\mu$ in terms of $2n$ Majorana operators $\gamma^{(a)}_\mu$ as follows:
\begin{equation}
\bar{\gamma}^{(a)}_\mu = \begin{cases}
\gamma^{(a)}_\mu \otimes \mathbb{1}^{(a)}_\Kcal, & \text{for } 1 \leq \mu \leq 2n, \\
\Gamma^{(a)}_\Hcal \otimes \gamma^{(a)}_{\mu - 2n}, & \text{for } 2n + 1 \leq \mu \leq 4n.
\end{cases}
\end{equation}
We decompose the Hilbert space as $(\Hcal \otimes \Kcal)^{\otimes k} \cong \Hcal^{\otimes k} \otimes \Kcal^{\otimes k}$, which is a simple permutation of the tensor factors. Then, we define the pair-parity operator $\bar{Q}_{ab}$, along with its $\Hcal$ and $\Kcal$ components $Q_{ab}^\Hcal$ and $Q_{ab}^\Kcal$:
\begin{equation}
\bar{Q}_{ab} = \bar{\Gamma}^{(a)} \bar{\Gamma}^{(b)} = \left( \Gamma^{(a)}_\Hcal \Gamma^{(b)}_\Hcal \right) \otimes \left( \Gamma^{(a)}_\Kcal \Gamma^{(b)}_\Kcal \right) = Q_{ab}^\Hcal \otimes Q_{ab}^\Kcal.
\end{equation}
This lets us write the bridge operators $\bar{\Lambda}_{ab}$ in terms of the bridge operators $\Lambda_{ab}$ acting on the $a$-th and $b$-th replicas of the $\Hcal^{\otimes k}$ and $\Kcal^{\otimes k}$ subsystems, respectively:
\begin{equation}
\bar{\Lambda}_{ab} = \sum_{\mu = 1}^{4n} \bar{\gamma}^{(a)}_\mu \bar{\gamma}^{(b)}_\mu = \Lambda^\Hcal_{ab} \otimes \mathbb{1}_{ab}^\Kcal + Q_{ab}^\Hcal \otimes \Lambda^\Kcal_{ab}.
\end{equation}
The parity operators satisfy several key properties:
\begin{itemize}
    \item As they act on different subsystems, all operators from $\{\Gamma^{(a)}_\Hcal, \Gamma^{(b)}_\Hcal, \Gamma^{(a)}_\Kcal, \Gamma^{(b)}_\Kcal\}$ commute with each other.
    \item The parity and pair-parity operators are all Hermitian and square to the identity.
    \item The parity operator $\Gamma^{(a)}_\Hcal$ and $\Gamma^{(b)}_\Hcal$ anticommute with the bridge operator $\Lambda_{ab}^\Hcal$, and similarly for $\Kcal$. This implies that the pair-parity operators satisfy:
    \begin{equation}
    [Q_{ab}^\Hcal, \Lambda_{ab}^\Hcal] = 0, \quad [Q_{ab}^\Kcal, \Lambda_{ab}^\Kcal] = 0.
    \end{equation} 
    \item Each vector $\bar{U} |\mathbf{0} \rangle$ has a definite parity $\pm 1$. Since $\GSym^k(\Hcal \otimes \Kcal)$ is spanned by states $(\bar{U} |\mathbf{0} \rangle)^{\otimes k}$ for $\bar{U} \in \mathfrak{M}_{2n}$, and $\bar{Q}_{ab}$ and $\bar{\rho}$ are both Hermitian, it follows that:
    \begin{equation} \label{eq:pair-parity-invariance}
    \bar{Q}_{ab} \bar{\rho} = \bar{\rho}, \quad \bar{\rho} \bar{Q}_{ab} = \bar{\rho}.
    \end{equation}
\end{itemize}
Writing out \eqref{eq:pair-parity-invariance} in full, we have that $(Q_{ab}^\Hcal \otimes Q_{ab}^\Kcal) \bar{\rho} = \bar{\rho}$. Multiplication of both sides by $(\mathbb{1}_{ab}^\Hcal \otimes Q_{ab}^\Kcal)$ yields the following shuffle identity:
\begin{equation}
    (Q_{ab}^\Hcal \otimes \mathbb{1}_{ab}^\Kcal) \bar{\rho} = (\mathbb{1}_{ab}^\Hcal \otimes Q_{ab}^\Kcal) \bar{\rho},
\end{equation}
and similarly for the left multiplication by $\bar{\rho}$. This lets us show that $[\rho, \Lambda_{ab}] = 0$ for all $1 \leq a < b \leq k$. Using the decomposition of the bridge operators $\bar{\Lambda}_{ab}$, we have:
\begin{align}
[\rho, \Lambda_{ab}] &= [\Tr_{\Kcal^{\otimes k}}(\bar{\rho}), \Lambda_{ab}] \\
&= \Tr_{\Kcal^{\otimes k}}([\bar{\rho}, \Lambda_{ab}^\Hcal \otimes \mathbb{1}_{ab}^\Kcal]) \\
& = \Tr_{\Kcal^{\otimes k}}( [\bar{\rho}, \bar{\Lambda}_{ab}] - [\bar{\rho}, Q_{ab}^\Hcal \otimes \Lambda_{ab}^\Kcal] ) \\
&= \Tr_{\Kcal^{\otimes k}}( (Q_{ab}^\Hcal \otimes \Lambda_{ab}^\Kcal)\bar{\rho} - \bar{\rho} (Q_{ab}^\Hcal \otimes \Lambda_{ab}^\Kcal) ),
\end{align}
where we have used the fact that $[\bar{\rho}, \bar{\Lambda}_{ab}] = 0$. Finally, using the shuffle identity, $[Q_{ab}^\Kcal, \Lambda_{ab}^\Kcal] = 0$, Lemma \ref{lem:partial-trace-cyclicity} and linearity of the trace we obtain:
\begin{align}
[\rho, \Lambda_{ab}] &= \Tr_{\Kcal^{\otimes k}}( (\mathbb{1}_{ab}^\Hcal \otimes \Lambda_{ab}^\Kcal Q_{ab}^\Kcal)\bar{\rho} - \bar{\rho} (\mathbb{1}_{ab}^\Hcal \otimes Q_{ab}^\Kcal \Lambda_{ab}^\Kcal) ) \\
&= \Tr_{\Kcal^{\otimes k}}( (\mathbb{1}_{ab}^\Hcal \otimes  Q_{ab}^\Kcal \Lambda_{ab}^\Kcal)\bar{\rho} - \bar{\rho} (\mathbb{1}_{ab}^\Hcal \otimes Q_{ab}^\Kcal \Lambda_{ab}^\Kcal) ) \\
&= \Tr_{\Kcal^{\otimes k}}( (\mathbb{1}_{ab}^\Hcal \otimes  Q_{ab}^\Kcal \Lambda_{ab}^\Kcal)\bar{\rho} - (\mathbb{1}_{ab}^\Hcal \otimes Q_{ab}^\Kcal \Lambda_{ab}^\Kcal) \bar{\rho} ) \\
&= 0.
\end{align}
By Lemma \ref{lem:parity-permutation-commutation}, it follows that $[\rho, Q_{ab}] = 0$ and $[\rho, \Pi] = 0$ for all $1 \leq a < b \leq k$ and $\Pi \in S_k$, which completes the proof that $(2) \Rightarrow (1)$.

We now prove that $(1) \Rightarrow (3)$. We only assume that $[\rho, \Lambda_{ab}] = 0$. Let $| \Omega \rangle \in \Hcal \otimes \Kcal$, where $\Kcal \cong \Hcal$, be a maximally entangled state in the computational basis:
\begin{equation}
    |\Omega \rangle = \frac{1}{\sqrt{2^n}} \sum_{\mathbf{x} \in \{0, 1\}^n} |\mathbf{x}\rangle_\Hcal |\mathbf{x}\rangle_\Kcal.
\end{equation}
Furthermore, define the $n$-qubit reversal permutation $\tilde{\Pi}$, which acts on $\Hcal$ as $\tilde{\Pi} |x_1 x_2 \cdots x_n \rangle = |x_n x_{n-1} \cdots x_1 \rangle$ (and similarly on $\Kcal$). Clearly, $\tilde{\Pi}^T = \tilde{\Pi}^{-1} = \tilde{\Pi}$, which lets us define the partially-reversed maximally entangled state:
\begin{equation}
|\tilde{\Omega} \rangle = (\tilde{\Pi} \otimes \mathbb{1}) | \Omega \rangle = (\mathbb{1} \otimes \tilde{\Pi}) | \Omega \rangle.
\end{equation} 
For any matrix $M$ acting on one half of this state, we have the following identity:
\begin{equation}
(\mathbb{1} \otimes M) | \tilde{\Omega} \rangle = ((M \tilde{\Pi})^T \otimes \mathbb{1}) | \Omega \rangle = (\tilde{\Pi} M^T \tilde{\Pi} \otimes \mathbb{1}) | \tilde{\Omega} \rangle.
\end{equation}
By standard results, we know that the density operator $\rho$ supported on $\Hcal^{\otimes k}$ admits a purification $|\bar{\rho} \rangle \in (\Hcal \otimes \Kcal)^{\otimes k} \cong \Hcal^{\otimes k} \otimes \Kcal^{\otimes k}$, given by:
\begin{equation}
|\bar{\rho} \rangle \propto (\sqrt{\rho}_{\Hcal^{\otimes k}} \otimes \mathbb{1}_\Kcal^{\otimes k}) | \tilde{\Omega} \rangle^{\otimes k}. 
\end{equation}
All that remains to be shown is that $|\bar{\rho} \rangle$ is Gaussian-symmetric, i.e. $\bar{\Lambda}_{ab} |\bar{\rho} \rangle = 0$ for all $a < b$. Since $[\rho, \Lambda_{ab}] = 0$, it follows that $[\sqrt{\rho}, \Lambda_{ab}] = 0$, and similarly $[\sqrt{\rho}, Q_{ab}] = 0$ by Lemma \ref{lem:parity-permutation-commutation}. Using the decomposition of the bridge operators $\bar{\Lambda}_{ab}$, we have:
\begin{align}
\bar{\Lambda}_{ab} |\bar{\rho} \rangle &\propto (\Lambda^\Hcal_{ab} \otimes \mathbb{1}_{ab}^\Kcal + Q_{ab}^\Hcal \otimes \Lambda^\Kcal_{ab}) (\sqrt{\rho}_{\Hcal^{\otimes k}} \otimes \mathbb{1}_\Kcal^{\otimes k}) | \tilde{\Omega} \rangle^{\otimes k} \\
&\propto (\sqrt{\rho}_{\Hcal^{\otimes k}} \otimes \mathbb{1}_\Kcal^{\otimes k}) (\Lambda^\Hcal_{ab} \otimes \mathbb{1}_{ab}^\Kcal + Q_{ab}^\Hcal \otimes \Lambda^\Kcal_{ab}) | \tilde{\Omega} \rangle^{\otimes k}.
\end{align}
Therefore, we need to show that $\bar{\Lambda}_{ab} | \tilde{\Omega} \rangle^{\otimes k} = 0$. Since $\bar{\Lambda}_{ab}$ acts trivially on all replicas except for the $a$-th and $b$-th, we may restrict our attention to the two-replica space $(\Hcal \otimes \Kcal)^{\otimes 2}$. We consider the second term:
\begin{equation}
(Q_{\Hcal_a \Hcal_b} \otimes \Lambda_{\Kcal_a \Kcal_b}) | \tilde{\Omega} \rangle_{\Hcal_a \Kcal_a}| \tilde{\Omega} \rangle_{\Hcal_b \Kcal_b}.
\end{equation}
Alternatively, we may write it as:
\begin{align}
    \left(\sum_{\mu = 2n+1}^{4n} \bar{\gamma}^{(a)}_\mu \bar{\gamma}^{(b)}_\mu \right) | \tilde{\Omega} \rangle_{\Hcal_a \Kcal_a} | \tilde{\Omega} \rangle_{\Hcal_b \Kcal_b} 
    &=  \sum_{\mu = 1}^{2n} \left( (\Gamma^{(a)} \otimes \gamma^{(a)}_{\mu}) \otimes (\Gamma^{(b)} \otimes \gamma_\mu^{(b)}) \right) | \tilde{\Omega} \rangle_{\Hcal_a \Kcal_a} | \tilde{\Omega} \rangle_{\Hcal_b \Kcal_b} \\
    &= \sum_{\mu = 1}^{2n} \left( \Gamma^{(a)} \otimes \gamma^{(a)}_{\mu} | \tilde{\Omega} \rangle_{\Hcal_a \Kcal_a} \right) \otimes \left( \Gamma^{(b)} \otimes \gamma_\mu^{(b)} | \tilde{\Omega} \rangle_{\Hcal_b \Kcal_b} \right).
\end{align}
From the properties of $| \tilde{\Omega} \rangle$, for both replicas we may bring the Majorana operators $\gamma_\mu$ to the other side of the tensor product:
\begin{equation}
    (\Gamma \otimes \gamma_\mu) | \tilde{\Omega} \rangle = (\Gamma \tilde{\Pi} \gamma_\mu^T \tilde{\Pi} \otimes \mathbb{1}) | \tilde{\Omega} \rangle.
\end{equation}
Working explicitly in the Jordan--Wigner representation, we have that $\gamma_\mu^T = (-1)^{\mu + 1} \gamma_\mu$. Now consider the action of the reversal permutation $\tilde{\Pi}$ on the Majorana operators $\gamma_\mu$. In this representation, the order of the qubits is reversed into the following form:
\begin{align}
    &\tilde{\Pi} \gamma_1 \tilde{\Pi} = \mathbb{1}^{\otimes n-1} \otimes X = i (Z^{\otimes n}) (\gamma_{2n})\\
    &\tilde{\Pi}\gamma_2 \tilde{\Pi} = \mathbb{1}^{\otimes n-1} \otimes Y = -i (Z^{\otimes n}) (\gamma_{2n-1})\\
    & \quad \vdots \\
    &\tilde{\Pi} \gamma_{2n-1} \tilde{\Pi} = X \otimes Z^{\otimes n-1} = i (Z^{\otimes n}) (\gamma_{2})\\
    &\tilde{\Pi} \gamma_{2n} \tilde{\Pi} = Y \otimes Z^{\otimes n-1} = -i (Z^{\otimes n}) (\gamma_{1}).
\end{align}
Using the fact that $\Gamma = Z^{\otimes n}$, this lets us write: $\tilde{\Pi} \gamma_\mu \tilde{\Pi} = i(-1)^{\mu + 1} \Gamma \gamma_{2n + 1 - \mu}$. Putting everything together, we have that:
\begin{align}
    (\Gamma \otimes \gamma_\mu) | \tilde{\Omega} \rangle &= (\Gamma \tilde{\Pi} \gamma_\mu^T \tilde{\Pi} \otimes \mathbb{1}) | \tilde{\Omega} \rangle \\
    &= (-1)^{\mu + 1} (\Gamma \tilde{\Pi} \gamma_\mu \tilde{\Pi} \otimes \mathbb{1}) | \tilde{\Omega} \rangle \\
    &= i (-1)^{2\mu + 2} (\Gamma^2 \gamma_{2n + 1 - \mu} \otimes \mathbb{1}) | \tilde{\Omega} \rangle \\
    &= i (\gamma_{2n + 1 - \mu} \otimes \mathbb{1}) | \tilde{\Omega} \rangle.
\end{align}
Therefore, the second term of $\bar{\Lambda}_{ab} | \tilde{\Omega} \rangle^{\otimes k}$ becomes:
\begin{align}
    \sum_{\mu = 1}^{2n} \left( \Gamma^{(a)} \otimes \gamma^{(a)}_{\mu} | \tilde{\Omega} \rangle_{\Hcal_a \Kcal_a} \right) \otimes \left( \Gamma^{(b)} \otimes \gamma_\mu^{(b)} | \tilde{\Omega} \rangle_{\Hcal_b \Kcal_b} \right) 
    &= \sum_{\nu = 1}^{2n} \left( i \gamma_\nu^{(a)} \otimes \mathbb{1}^{(a)} | \tilde{\Omega} \rangle_{\Hcal_a \Kcal_a} \right) \otimes \left( i \gamma_\nu^{(b)} \otimes \mathbb{1}^{(b)}  | \tilde{\Omega} \rangle_{\Hcal_b \Kcal_b} \right) \\
    &= - ( \Lambda_{\Hcal_a \Hcal_b} \otimes \mathbb{1}_{\Kcal_a \Kcal_b}) | \tilde{\Omega}  \rangle_{\Hcal_a \Kcal_a}| \tilde{\Omega}  \rangle_{\Hcal_b \Kcal_b},
\end{align}
where we have re-indexed the sum as $\nu = 2n + 1 - \mu$. Overall, the first and second terms of $\bar{\Lambda}_{ab} | \tilde{\Omega} \rangle^{\otimes k}$ cancel:
\begin{align}
\bar{\Lambda}_{ab} | \tilde{\Omega}  \rangle_{\Hcal_a \Kcal_a}| \tilde{\Omega}  \rangle_{\Hcal_b \Kcal_b} 
&= (\Lambda_{\Hcal_a \Hcal_b} [\mathbb{1}_{\Hcal_a \Hcal_b} - \mathbb{1}_{\Hcal_a \Hcal_b}] \otimes \mathbb{1}_{\Kcal_a \Kcal_b}) | \tilde{\Omega} \rangle_{\Hcal_a \Kcal_a}| \tilde{\Omega} \rangle_{\Hcal_b \Kcal_b} \\
&= 0.
\end{align}
Therefore, it follows that $\bar{\Lambda}_{ab} | \tilde{\Omega} \rangle^{\otimes k} = 0$, which implies that $\bar{\Lambda}_{ab} |\bar{\rho} \rangle = 0$ for all $1 \leq a < b \leq k$. Hence, $|\bar{\rho} \rangle$ is a purification of $\rho$ in $\GSym^k(\Hcal \otimes \Kcal)$, completing the proof that $(1) \Rightarrow (3)$. 

Finally, $(3) \Rightarrow (2)$ follows by setting $\rho = \Tr_{\Kcal^{\otimes k}}(|\bar{\rho} \rangle \langle \bar{\rho}|)$ for any $|\bar{\rho} \rangle \in \GSym^k(\Hcal \otimes \Kcal)$.
\end{proof}

\bibliography{bibliography.bib}

\end{document}